\documentclass[10pt,twocolumn,twoside]{IEEEtran}

\usepackage{multirow}
\usepackage[T1]{fontenc}
\usepackage{picinpar}
\usepackage{url}
\usepackage{colortbl}
\usepackage{soul}
\usepackage{float}
\usepackage{authblk}
\usepackage{mathrsfs}
\usepackage{subfigure}
\usepackage{array}
\usepackage{booktabs}
\usepackage{epsfig}
\usepackage{color}
\usepackage{color}
\usepackage{graphics}
\usepackage{amsfonts} % for pdf, bitmapped graphics files
\usepackage{epsfig} % for postscript graphics files
\usepackage{amsmath} % assumes amsmath package installed
\usepackage{amssymb}  % assumes amsmath package installed
\usepackage{amsthm}  %提供对定理的排版。加了proof环境就不报错
\usepackage{hyperref}  %超链接引用
\usepackage{cite}
\usepackage{upgreek}       
\usepackage[ruled,linesnumbered,vlined]{algorithm2e}
\graphicspath{{figures/}}%定义图片存放的路径    

\newtheorem{theorem}{\bf Theorem}
\newtheorem{lemma}{\bf Lemma}
\newtheorem{remark}{\bf Remark}

\newtheorem{assumption}{\bf Assumption}

\begin{document}
	\title{Composite Learning Adaptive Control \textcolor{blue}{under Non-Persistent Partial Excitation} }
	\author{
		Jiajun Shen, Wei Wang, \IEEEmembership{Senior Member, IEEE}, Changyun Wen, \IEEEmembership{Fellow, IEEE}, Jinhu L\"u, \IEEEmembership{Fellow, IEEE}
		\thanks{This work was supported in part by National Natural Science Foundation of China under Grant 62373019, in part by Fundamental Research Funds for the Central Universities under Grants TKF-2025010466166 and JKF-2025033060898, and in part by Academic Excellence Foundation of BUAA for PhD Students.
		\textit{(Corresponding author: Jinhu L\"u.)}}
		\thanks{Jiajun Shen, Wei Wang and Jinhu L\"u are with the School of Automation Science and Electrical Engineering, Beihang University, Beijing 100191, P.R. China. Wei Wang and Jinhu L\"u are also with the Hangzhou Innovation Institute, Beihang University, Hangzhou 310051, P.R. China, and the Zhongguancun Laboratory, Beijing 100194, P.R. China. 
		(e-mail: jjshen@buaa.edu.cn; w.wang@buaa.edu.cn; jhlu@iss.ac.cn).}
		\thanks{Changyun Wen is with the School of Electrical and Electronic Engineering, Nanyang Technological University, Nanyang Avenue, 639798, Singapore. (e-mail: ecywen@ntu.edu.sg).}
	}
	\maketitle

	\begin{abstract}
	This paper focuses on \textcolor{blue}{relaxing the excitation conditions for the} adaptive control \textcolor{blue}{of} uncertain nonlinear systems.
	By adopting the spectral decomposition technique, a linear regression equation \textcolor{blue}{(LRE)} is constructed to \textcolor{blue}{quantitatively} collect \textcolor{blue}{historical} excitation information, \textcolor{blue}{based on which the parameter estimation error is decomposed into the excited component and the unexcited component.}
	\textcolor{blue}{By sufficiently utilizing the collected excitation information, the composite learning and $\mu$-modification terms are designed and incorporated} into the ``Lyapunov-based'' parameter update law.
	\textcolor{blue}{By developing a novel Lyapunov function, it} is demonstrated that \textcolor{blue}{under non-persistent partial excitation, the control error and the} excited parameter estimation error component converge to zero, while the unexcited component remains \textcolor{blue}{bounded}.
	\textcolor{blue}{Furthermore, the proposed adaptive control scheme can effectively eliminate the effects of parametric uncertainties and enhance the robustness of the closed-loop systems.}
	Simulation results are provided to verify the theoretical findings.
	\end{abstract}

	\section{Introduction}
	Adaptive control is an important discipline in the field of control, which has evolved over several decades.
	Numerous valuable adaptive control schemes have been proposed, which are elaborated in detail in \cite{goodwin1984adaptive,krstic1995nonlinear,ioannou1996robust,aastrom2008adaptive,narendra1989stable,sastry1989adaptive,tao2003adaptive,zhou2008adaptive,ortega2020modified,guo2023composite} and the references therein.
	\textcolor{blue}{Excitation information is an important concept in adaptive control systems, which reflects the effects of parametric uncertainties on system states and serves as the data basis for parameter estimation.
	A substantial number of extant achievements necessitate specific excitation conditions to compensate for the uncertainties, including persistent excitation (PE) \cite{sastry1989adaptive}, sufficient excitation (SE) \cite{krause1987parameter}, and interval excitation (IE) \cite{kreisselmeier1990richness}.
	Specifically, we denote $\phi(x) \in \mathbb{R}^{p \times n}$ as the nonlinear regressor for an unknown constant vector, it is said to satisfy PE / SE / IE conditions if there exist constants $\tau_d, \rho, t_e \in \mathbb{R}^+$ such that
	\begin{align}
		\int_{t-\tau_d}^{t} \phi(x(\tau)\big) \phi(x(\tau)\big)^T d\tau &\geq \rho I_p, \text{ }  
		\forall t \geq 0,   \tag{PE} \label{PE} \\
		\int_{0}^{t} \phi(x(\tau)\big) \phi(x(\tau)\big)^T d\tau &\geq \rho I_p,   \text{ }  
		\forall t \geq t_e,   \tag{SE} \label{SE} \\
		\int_{t_e-\tau_d}^{t_e} \phi(x(\tau)\big) \phi(x(\tau)\big)^T d\tau &\geq \rho I_p.   
		\tag{IE} \label{IE}
	\end{align}
	It is crucial to emphasize that these excitation conditions are challenging to ensure and verify in advance, as they rely on the future system states. 
	Although the exogenous excitation input method \cite{boyd1986necessary} can ensure these excitation conditions, it may lead to excessive energy consumption and unsatisfactory control performance. In what follows, the achievements of adaptive control are comprehensively reviewed from the perspective of relaxing excitation conditions.}
	
	\subsection{Literature Review}
	\textcolor{blue}{Most of the classical adaptive control methods, e.g. standard ``Lyapunov-based'' \cite{krstic1995nonlinear} and ``estimation-based'' \cite{ioannou1996robust} approaches, can eliminate the effects of uncertainties completely if and only if the regressor satisfies the PE conditions. The $\sigma$-modification \cite{ioannou1983adaptive} and $e$-modification \cite{narendra1987a} schemes are presented to construct the negative definite terms of the parameter estimation errors. Although these schemes relax the excitation conditions, a positive term is introduced into the derivative of Lyapunov function, thereby reducing the control performance. The projection operator methods \cite{krstic1995nonlinear} guarantee the boundedness of parameter estimates without imposing excitation conditions, while they require prior knowledge of the unknown parameters and cannot ensure the convergence of parameter estimation. The composite control schemes \cite{slotine1989composite} combine the direct and indirect adaptive laws to improve the estimation performance, while the PE conditions are still required.}
	
	As introduced in the surveys \cite{ortega2020modified,guo2023composite}, several data-driven adaptive control schemes have been proposed to relax the \textcolor{blue}{excitation conditions by utilizing historical excitation information.}
	\textcolor{blue}{The main achievements include} concurrent learning \cite{chowdhary2010concurrent,chowdhary2013concurrent,chowdhary2014exponential,li2022concurrent,long2023filtering,parikh2019integral}, composite learning \cite{pan2016composite,pan2018composite,roy2018combined,pan2017composite,guo2019composite,guo2020composite,cho2018composite,pan2019efficient,pan2022bioinspired}, regulation-triggered batch identifier  \cite{karafyllis2018adaptive,karafyllis2019adaptive,karafyllis2020adaptive,shen2025adaptive} and other approaches \cite{adetola2008finite,ortega2021new,ortega2021parameter}.
	\textcolor{blue}{The concurrent learning and composite learning methods ensure the exponential convergence of parameter estimation, thereby improving the control performance and enhancing the robustness of the closed-loop systems \cite{boyd1986necessary,anderson1977exponential,bodson1984exponential}.
	Next, let's review these techniques from the viewpoints of} excitation collection and excitation utilization.
	
	\textcolor{blue}{Given that the original concurrent learning schemes \cite{chowdhary2010concurrent,chowdhary2013concurrent,chowdhary2014exponential} are challenging to implement because they require the derivatives of the system states, we mainly focus on the integral-based concurrent learning schemes \cite{parikh2019integral}.
	Given an uncertain system $\dot{x}=f(x)+u+\phi(x)^T\theta$ with $\theta$ of an unknown constant vector,} these approaches employ excitation detection algorithms \cite{chowdhary2011singular} \textcolor{blue}{to determine the excitation collection intervals $(t_{i1},t_{i2})$, and construct the LRE $\sum_{i=1}^{N} \varPhi_i \varPi_i= \sum_{i=1}^{N}\varPhi_i \varPhi_i^T \theta$ by utilizing the following equations,
	\begin{align}
		\varPi_i &= x(t_{i2}) - x(t_{i1}) - \int_{t_{i1}}^{t_{i2}} \big( f(x(\tau)) + u(\tau) \big) d\tau,
		\label{V} \\
		\varPhi_i &= \int_{t_{i1}}^{t_{i2}} \phi (x(\tau)) d\tau.
		\label{Phi}
	\end{align} 
	However, some valuable excitation information may be lost due to the detection errors \cite{chowdhary2011singular} and the potential integral cancellation in (\ref{Phi}).}
	\textcolor{blue}{In comparison, the composite learning approaches \cite{pan2016composite} construct the LRE $Z_f(t)=W_f(t)\theta$ by utilizing the following ordinary differential equations (ODEs),
	\begin{align}
		\dot{Z}_f &= \phi_f(x) \varpi_f(x,u), \text{ } Y(0)=0,  \label{Y} \\
		\dot{W}_f &= \phi_f(x) \phi_f(x)^T, \text{ } Q(0)=0,  \label{Q} 
	\end{align} 
	where $\phi_f(x)$ and $\varpi_f(x,u)$ are the filtered signals of $\phi(x)$ and $\dot{x} - f(x) - u$, respectively.
	Note that the vector $Z_f$ and the excitation matrix $W_f$ may become unbounded due to the positive semidefinite term $\phi_f(x) \phi_f(x)^T$ in (\ref{Q}), which may cause numerical divergence. To address this issue,}
	several studies \cite{cho2018composite,pan2019efficient} incorporate forgetting factors to ensure \textcolor{blue}{the boundedness of the LRE}.
	\textcolor{blue}{Moreover, \cite{lee2019concurrent,goel2020recursive,lai2024sift} present directional forgetting schemes to attenuate the collected excitation information solely along the present excitation directions, thereby avoiding the unexpected attenuation of other excitation directions.
	However, these excitation forgetting approaches} have failed to consider the variation in richness of the collected excitation information across different directions.
	When the excitation information pertaining to strong \textcolor{blue}{excitation} directions is forgotten, the weak \textcolor{blue}{excitation} directions may be unexpectedly attenuated \textcolor{blue}{simultaneously}. 
	By utilizing the LRE, the concurrent learning and composite learning terms are designed to enhance the estimation and control performance.
 	\textcolor{blue}{It is demonstrated that the resulting closed-loop systems will be exponentially stable once the \ref{SE} and \ref{IE} conditions are satisfied.}
 	However, the \ref{SE} and \ref{IE} conditions are still restrictive in many practical applications.
 	
 	\textcolor{blue}{Noticing that} the main objective of adaptive control is to ensure control performance rather than \textcolor{blue}{estimation accuracy, several partial identification schemes \cite{bittanti1990recursive,bittanti1992effective,kogan1996locally,li1998geometric} were proposed.
 	The parameter estimation errors are decomposed into the persistently excited and the non-persistently excited components.
 	It is demonstrated that by utilizing these adaptive control schemes, the persistently excited component converges to zero.
 	Moreover, $\mu$-modification methods \cite{uzeda2023adaptive,uzeda2023robust,uzeda2025estimation} are presented to guarantee the boundedness of the non-persistently excited component and enhance the robustness of the closed-loop systems.
 	However, these approaches generally need restrictive partial persistent excitation conditions and cannot completely eliminate the effects of parametric uncertainties.}
 	
 	\subsection{Contributions}
 	\textcolor{blue}{Motivated by the discussions mentioned above, this paper presents a composite learning adaptive control scheme for uncertain systems under non-persistent partial excitation.
 	By employing the spectral decomposition technique, a LRE is constructed to quantitatively collect historical excitation information.} 
 	The parameter estimation error $\tilde{\theta}$ is orthogonally decomposed into the excited component $\tilde{\theta}_e$ and the unexcited component $\tilde{\theta}_u$, \textcolor{blue}{where $\tilde\theta_e$ affects control performance and $\tilde\theta_u$ does not. Specifically, we have $\exists \tau \in [0,t]$, $\phi (x(\tau))^T \tilde{\theta}_e(t) \neq 0$ and $\phi (x(\tau))^T \tilde{\theta}_u(t) = 0$, $\forall \tau \in [0,t]$.
	The parameter update law is designed based on composite learning and $\mu$-modification techniques, which sufficiently utilize the collected excitation information.}
	The contributions are mainly three-folds.
	
	\begin{itemize}  
		\item [$\bullet$] 
		In contrast to the existing excitation collection schemes \cite{chowdhary2010concurrent,parikh2019integral,pan2016composite,cho2018composite,pan2019efficient,lee2019concurrent,goel2020recursive,lai2024sift}, \textcolor{blue}{this paper employs} spectral decomposition technique to decompose the collected excitation information onto independent excitation directions and design independent forgetting factors. It is illustrated that all the spectra of \textcolor{blue}{historical} excitation information are \textcolor{blue}{sufficiently} collected and the boundedness of the LRE is \textcolor{blue}{ensured}.
		\item [$\bullet$] 
		Compared with the existing composite learning adaptive control schemes \cite{pan2016composite,pan2018composite,roy2018combined,pan2017composite,guo2019composite,guo2020composite,cho2018composite,pan2019efficient,pan2022bioinspired}, \textcolor{blue}{the excitation conditions are further relaxed.}
		\textcolor{blue}{It is demonstrated that under non-persistent partial excitation, the control error} and the excited parameter estimation error component converge exponentially to zero, and the unexcited parameter estimation error component \textcolor{blue}{remains bounded}.
		\textcolor{blue}{The effects of uncertainties can be completely eliminated and the robustness of the closed-loop systems is enhanced.}
		\item [$\bullet$]
		\textcolor{blue}{By incorporating the RBFNN and dynamic surface control techniques, the proposed control scheme is extended to high-order systems with unstructured uncertainties. It is demonstrated that the closed-loop system is semi-globally stable under non-persistent partial excitation.}
	\end{itemize}
	
	The rest of this paper is organized as follows. 
	\textcolor{blue}{In Section \uppercase\expandafter{\romannumeral2}, we formulate the problems and present the composite learning adaptive control scheme.}
	In Section \uppercase\expandafter{\romannumeral3}, we present a composite learning adaptive dynamic surface control scheme for high-order systems with unstructured uncertainties.
	In Section \uppercase\expandafter{\romannumeral4}, simulation results are provided \textcolor{blue}{to validate the theoretical findings}, followed by a conclusion drawn in Section \uppercase\expandafter{\romannumeral5}.
	
	\textbf{Notations.}
	The sets of real numbers and non-negative real numbers are denoted by $\mathbb{R}$ and $\mathbb{R}^+$, respectively.
	The set of positive integers is denoted by $\mathbb{Z}^+$.
	For a vector $x \in \mathbb{R}^n$, $\|x\|$ denotes the Euclidean norm of $x$.
	Given a matrix $W \in \mathbb{R}^{n \times n}$, ${\rm rank}[W]$ denotes the rank of $W$, $\mathcal{R}[W]$ and $\mathcal{N}[W]$ denote the range space and null space of $W$, respectively.
	Denote $\lambda_k$, $k=1,2,...,h$ as the $h$ distinct eigenvalues of $W$, the eigen-space corresponding to $\lambda_k$ is denoted by $\mathcal{E}(\lambda_k)$.
	\textcolor{blue}{$I_n \in \mathbb{R}^{n \times n}$ and $O^{m \times n}$ represent the identity matrix and the zero matrix, respectively.}
	Given a space $\mathcal{S} \subset \mathbb{R}^n$, ${\rm Proj}(x,\mathcal{S})$ denotes the orthogonal projection of $x$ on  $\mathcal{S}$, $\dim(\mathcal{S})$ denotes the dimension of $\mathcal{S}$.
	\textcolor{blue}{Given a signal $\zeta(t)$ that is discontinuous at moment $\tau$, its left and right values are denoted by $\zeta(\tau^-)$ and $\zeta(\tau^+)$, respectively.}

	\section{Composite Learning Adaptive Control}
	Consider a first-order uncertain system modeled as follows,
	\begin{equation}
		\dot{x} = f(x) + u + \phi(x)^T \theta, 
		\label{model}
	\end{equation}
	where $x \in \mathbb{R}^n$ and $u \in \mathbb{R}^n$ are the system state and control input, respectively.
	$f : \mathbb{R}^n \rightarrow \mathbb{R}^n$ is a known mapping, $\phi : \mathbb{R}^n \rightarrow \mathbb{R}^{p \times n}$ is a nonlinear regressor.
	$\theta \in \mathbb{R}^p$ denotes the vector of unknown constant parameters, the dimension $p \in \mathbb{Z}^+$ is the number of unknown parameters.
	
	\subsection{Problem Formulation and Preliminaries}
	Denote the reference signal of state $x$ as $x_r(t)$ and the estimate of unknown vector $\theta$ as $\hat{\theta}$. The tracing control error and the parameter estimation error are defined as follows,
	\begin{align}
		e &= x-x_r, \label{control-error}  \\
		\tilde{\theta} &= \theta - \hat{\theta}.  \label{estimation-error}
	\end{align} 
	
	It will be shown in later analysis that the parameter estimation error $\tilde{\theta}$ can be orthogonally decomposed into the excited component $\tilde{\theta}_e$ and the unexcited component $\tilde{\theta}_u$.
	The latter component $\tilde{\theta}_u$ is defined as the one \textcolor{blue}{does not affected the control performance,} in the sense that  $\phi (x(\tau))^T \tilde{\theta}_u(t) = 0$ is satisfied for all $\tau \in [0,t]$.
	\textcolor{blue}{Noticing that the main objective of adaptive control is to ensure control performance rather than estimation accuracy, the bounded unexcited component $\tilde{\theta}_u$ can generally be tolerated.}
	
	The \textcolor{blue}{\textit{design goals}} in this \textcolor{blue}{section} is to design a LRE and \textcolor{blue}{a composite learning} adaptive controller for uncertain \textcolor{blue}{nonlinear} system (\ref{model}) such that the following results can be obtained \textcolor{blue}{under non-persistent partial excitation.} 
	\begin{itemize}
		\item [$\bullet$]
		All the spectra of \textcolor{blue}{historical} excitation information are sufficiently collected into the LRE.
		\item [$\bullet$]
		\textcolor{blue}{All the signals in the closed-loop system remain bounded.}
		\item [$\bullet$]
		\textcolor{blue}{The tracking control error (\ref{control-error}) and the excited parameter estimation error component $\tilde{\theta}_e$ converge exponentially to zero, i.e., $\lim_{t \rightarrow \infty} e(t) \rightarrow O^{n \times 1}$, $\lim_{t \rightarrow \infty} \tilde{\theta}_e(t) \rightarrow O^{p \times 1}$.}
	\end{itemize}
	
	To achieve the \textcolor{blue}{design goals}, we introduce several lemmas and impose a necessary assumption \textcolor{blue}{in the following}, which will play essential roles in control design and stability analysis.
	To avoid distracting the readers from the main content, the proof of Lemma 3 is given in Appendix.
	\begin{lemma}
		\hspace{-0.25cm} \cite{lancaster1985theory}
		Consider a normal matrix $W \in \mathbb{R}^{p \times p}$ that has $h$ distinct eigenvalues $\lambda_k$, $k=1,2,...,h$.
		Denote $e_{k,l}$, $l=1,2,...,k_l$ as a group of unit orthogonal basis of the eigen-space corresponding to the eigenvalue $\lambda_k$.
		The spectral matrix $E_k \in \mathbb{R}^{p \times p}$ is defined as $E_k=\sum_{l=1}^{k_l} e_{k,l}e_{k,l}^T$, then \textcolor{blue}{the spectral decomposition form can be expressed as} $W=\sum_{k=1}^{h} \lambda_k E_k$.
		The set of spectral matrices $\left\{ E_1,E_2,...,E_h \right\}$ is unique \textcolor{blue}{and has} the following properties,
		\begin{align}
			\sum_{k=1}^{h} E_k &= I, \label{prop_0}  \\
			E_k &= E_k^T=\left(E_k\right)^2, \text{ }  \forall k=1,2,...,h, \label{prop_1} \\
			E_{k_1} E_{k_2} &= 0^{p \times p}, \text{ }  \forall k_1,k_2=1,2,...,h, \text{ } k_1 \neq k_2. \label{prop_2} 
		\end{align}
		From (\ref{prop_1}), the spectral matrices $E_k$, $k=1,2,...,h$ are orthogonal projection (idempotent) matrices.
	\end{lemma}
	
	\begin{lemma}
		\hspace{-0.25cm} \cite{horn2013matrix}
		For a given matrix $W \in \mathbb{R}^{p \times p}$, the range space $\mathcal{R}[W]$ is the orthogonal complement of the null space $\mathcal{N}[W^T]$, i.e., $\mathcal{R}[W] = \mathcal{N}[W^T]^{\bot}$.
	\end{lemma}
	
	\begin{lemma}
		\textcolor{blue}{Given} a vector $\nu \in \mathbb{R}^p$ and a positive semidefinite matrix $W \in \mathbb{R}^{p \times p}$ \textcolor{blue}{that has} $h$ distinct eigenvalues \textcolor{blue}{$\lambda_k$, $k=1,2,...,h$, where $\lambda_1$ denotes} the zero eigenvalue.
		If the vector $\nu$ satisfies ${\rm Proj} \left( \nu, \mathcal{E}(\lambda_1) \right)=0$, the inequality $\nu^T W \nu \geq \lambda_{\min}^+(W)\nu^T\nu$ \textcolor{blue}{holds}, where $\lambda_{\min}^+(W)$ represents the smallest positive eigenvalue of $W$.
	\end{lemma}
	
	\begin{assumption}
		The	reference signal $x_r(t)$ and its derivative are piecewise continuous and bounded.
	\end{assumption}
	
	\textcolor{blue}{This section presents a novel composite learning adaptive control framework, as shown in Fig. \ref{liucheng}. The LRE quantitatively collects historical excitation information, the parameter update law utilizes the previously collected excitation information to update the parameter estimate $\hat{\theta}$, the certainty equivalence (CE) control law utilizes the parameter estimate to generate the control input signal.}
	\begin{figure}[!htbp]
		\vspace{-0.2cm}
		\centering
		\includegraphics[width=8.75cm]{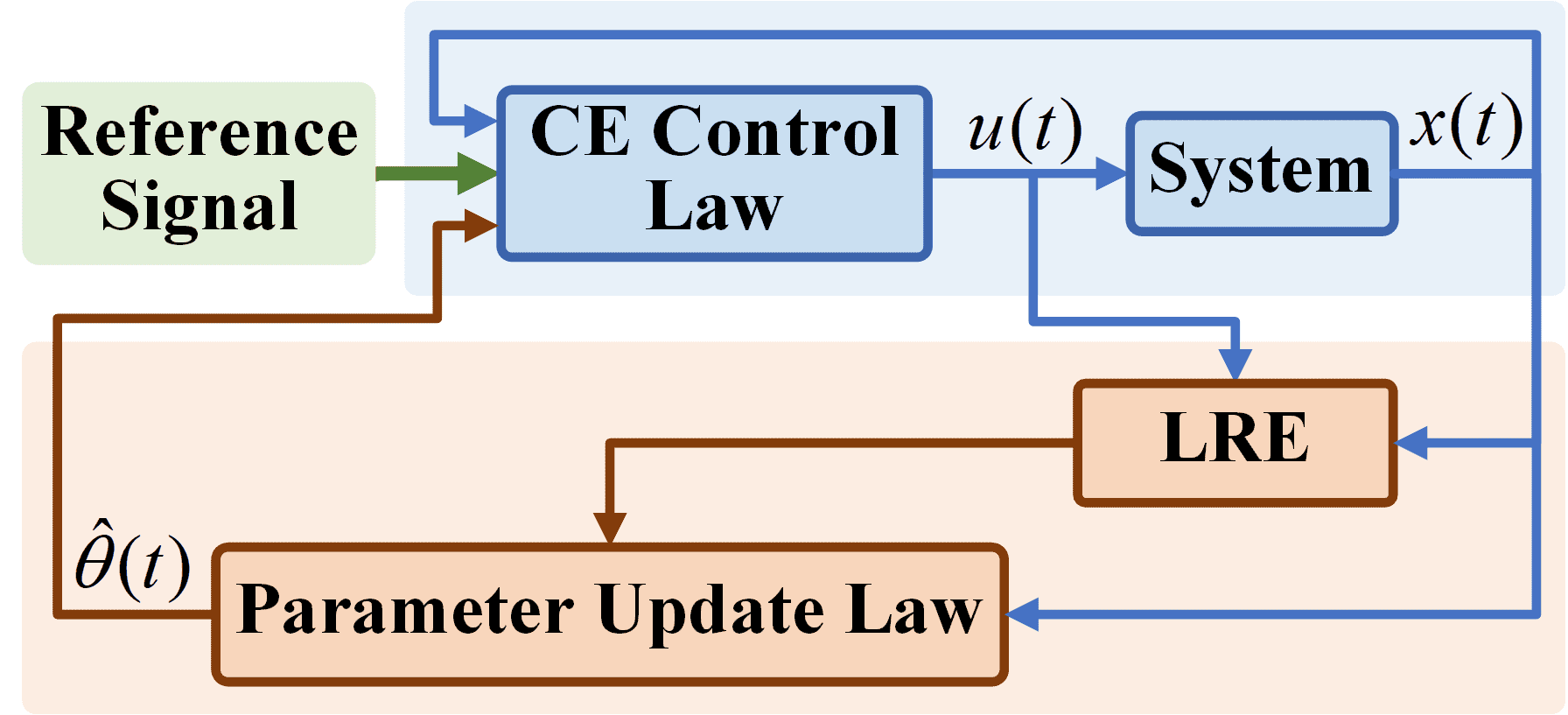}
		%\vspace{-0.2cm}
		\caption{\textcolor{blue}{Illustrating figure for the proposed composite learning adaptive control framework.}}
		\label{liucheng}
		\vspace{-0.2cm}
	\end{figure}

	\subsection{Design and Analysis of Linear Regression Equation}
	\textcolor{blue}{The LRE will be constructed as $Z(t)=W(t)\theta$, where the vector $Z(t) \in \mathbb{R}^p$ and the excitation matrix $W(t) \in \mathbb{R}^{p \times p}$ will be designed subsequently.}
	\textcolor{blue}{We will utilize the spectral decomposition technique to} orthogonally decompose the \textcolor{blue}{LRE and the} previously collected excitation information onto several \textcolor{blue}{excitation} directions \textcolor{blue}{according to excitation richness}.
	\textcolor{blue}{Then the LRE will be designed to} collect the present excitation information while weakening the previously collected excitation \textcolor{blue}{information} of some \textcolor{blue}{excitation} directions, on which the excitation is already \textcolor{blue}{sufficiently rich} for parameter estimation.

	The regressor $\phi(x)$ \textcolor{blue}{is generally regarded as the} carrier of excitation information \textcolor{blue}{in system identification and adaptive control techniques}.
	\textcolor{blue}{Moreover, the unknown vector $\theta$ is identifiable if and only if it belongs to the range space of the regressor.}
 	Similar to the existing excitation collection methods \cite{chowdhary2010concurrent,parikh2019integral,pan2016composite,cho2018composite,pan2019efficient,lee2019concurrent,goel2020recursive,lai2024sift}, the \textcolor{blue}{excitation} matrix $W(t)$ will collect the spectra of the regressor.
	\textcolor{blue}{For easier introduction of the design procedure, $W(t)$ is provisionally} defaulted to be positive semidefinite.
	Note that this is not an additional assumption, \textcolor{blue}{and it will be rigorously proved once we complete the design of LRE.}
	\textcolor{blue}{Before presenting the expressions of $Z(t)$ and $W(t)$, we first introduce}
	the transformed system model and the decomposed LRE.
	\\
	\textbf{Transformed system model:}
	\textcolor{blue}{By subtracting $\big(f(x)+u\big)$ from both sides of (\ref{model}) and} multiplying both sides on the left by $\phi(x)$, \textcolor{blue}{the transformed system model is obtained as follows,}
	\begin{equation}
		\varpi(\dot{x},x,u) = \phi(x) \phi(x)^T\theta,
		\label{transfor_1}
	\end{equation}
	where $\varpi(\dot{x},x,u)=\phi(x) \left(\dot{x}-f(x) - u\right)$ is defined for easier formulation of the design procedure.
	\\
	\textbf{Decomposed LRE:}
	\textcolor{blue}{For a specific moment $t$, denote the LRE as $Z(t)=W(t) \theta$.}
	From Lemma 1, the spectral decomposition form of \textcolor{blue}{the excitation} matrix $W(t)$ is given as follows,
	\begin{equation}
		W(t)=\sum_{k=1}^{h(t)} \lambda_k(t) E_k(t). \label{LRE_2}
	\end{equation}
	\textcolor{blue}{ Noticing that} the spectral matrices $E_k(t)$, $k=1,2,...,h(t)$ are constructed by the unit orthogonal eigenvectors of different eigenvalues, \textcolor{blue}{they represent} the \textcolor{blue}{independent excitation} directions of the previously collected excitation information.
	Moreover, the eigenvalues \textcolor{blue}{represent the excitation richness of these excitation} directions.
	Specially, the spectral matrix of the zero eigenvalue represents the \textcolor{blue}{excitation} direction, on which no excitation information has been collected yet.
	From (\ref{prop_0}) and (\ref{LRE_2}), the LRE is transformed to the following form,
	\begin{equation}
		Z(t) = \sum_{k=1}^{h(t)} E_k(t) Z(t) 
			 = \sum_{k=1}^{h(t)} \lambda_k(t) E_k(t) \theta.
		\label{decomposition}
	\end{equation}
	\textcolor{blue}{Noticing the orthogonal property of the spectral matrices} in (\ref{prop_2}), the transformed LRE (\ref{decomposition}) can be decomposed onto the $h(t)$ independent excitation directions, as shown below,
	\begin{equation}
		E_k(t) Z(t) = \lambda_k(t) E_k(t) \theta, \text{ } k=1,2,...,h(t).    \label{direction}
	\end{equation}
	\textbf{Design of the LRE:}
	According to the transformed system model (\ref{transfor_1}) and the decomposed LRE (\ref{direction}), we have
	\begin{equation}
		\begin{aligned}
			&\varpi(\dot{x},x,u) - \sum_{k=1}^{h(t)} \textcolor{blue}{\frac{\beta(\lambda_k,x)}{\lambda_k(t)} E_k(t) Z(t)} \\
			=& \left(\phi(x) \phi(x)^T- \sum_{k=1}^{h(t)} \textcolor{blue}{\beta(\lambda_k,x) E_k(t)} \right) \theta,
		\end{aligned}
		\label{transfor_2}
	\end{equation}
	where $\beta: \mathbb{R}^+ \times \mathbb{R}^n \rightarrow \mathbb{R}$ is the forgetting factor function, which will be designed later.
	Integrating (\ref{transfor_2}) \textcolor{blue}{over} the interval $[0,t]$, we obtain the following LRE,
	\begin{equation}
		Z(t) = W(t) \theta,
		\label{transfor_9}
	\end{equation}
	where the vector $Z(t)$ and \textcolor{blue}{the excitation matrix} $W(t)$ can be obtained \textcolor{blue}{in real time} by employing the following ODEs,
	\begin{align}
		\dot{Z} &= \varpi(\dot{x},x,u) - \sum_{k=1}^{h(t)} \textcolor{blue}{\frac{\beta(\lambda_k,x)}{\lambda_k(t)} E_k(t) Z(t)}, \label{Z} \\
		\dot{W} &= \phi(x) \phi(x)^T- \sum_{k=1}^{h(t)} \textcolor{blue}{\beta(\lambda_k,x) E_k(t)},  \label{W} 
	\end{align}
	where the initial values are set as $Z(0) = O^{p \times 1}$ and $W(0) = O^{p \times p}$.
	The first terms in (\ref{Z}) and (\ref{W}) collect the present excitation information, while the second terms are added to guarantee the boundedness of $Z(t)$ and $W(t)$.
	\\
	\textbf{Forgetting factor:}
	Inspired by the related results \cite{cho2018composite,pan2019efficient}, the forgetting factor function $\beta(\lambda_k,x)$ is designed as follows,
	\begin{equation}
		\beta(\lambda_k,x)=
		\begin{cases}
			0, \text{ } {\rm if} \text{ } \lambda_k \leq \sigma_{\min},  \\
			\frac{\lambda_{\max}\left(\phi(x)\phi(x)^T\right)}{2} 
			\left( {\rm sat} \left( \frac{\lambda_k-\mu}{\omega} \right)+1 \right), {\rm otherwise},
		\end{cases}
		\label{forgetting}
	\end{equation}
	where $\lambda_{\max}\left(\phi(x)\phi(x)^T\right)$ denotes the maximum eigenvalue of the matrix $\phi(x)\phi(x)^T \in \mathbb{R}^{p \times p}$. 
	$\mu=\frac{\sigma_{\max}+\sigma_{\min}}{2}$ and $\omega=\frac{\sigma_{\max}-\sigma_{\min}}{2}$ represent the center and width of the saturation function, respectively.
	Here, $\sigma_{\min} \in \mathbb{R}^+$ and $\sigma_{\max} \in \mathbb{R}^+$ are \textcolor{blue}{design parameters, the selection of which will be discussed after we complete the design and analysis.}
	The saturation function is presented as follows,
	\begin{equation}
		{\rm sat}(y)= \begin{cases}
			-1, &\text{if } y \leq -1, \\
			y,  &\text{if } -1<y<1, \\
			1,  &\text{if } y \geq 1.
		\end{cases} \label{saturation}
	\end{equation}
	The relationship between the \textcolor{blue}{forgetting factor} $\beta(\lambda_k,x)$ and the eigenvalue $\lambda_k$ is illustrated in Fig \ref{factor}.
	
	\begin{figure}[!htbp]
		\vspace{-0.2cm}
		\centering
		\includegraphics[width=7.5cm]{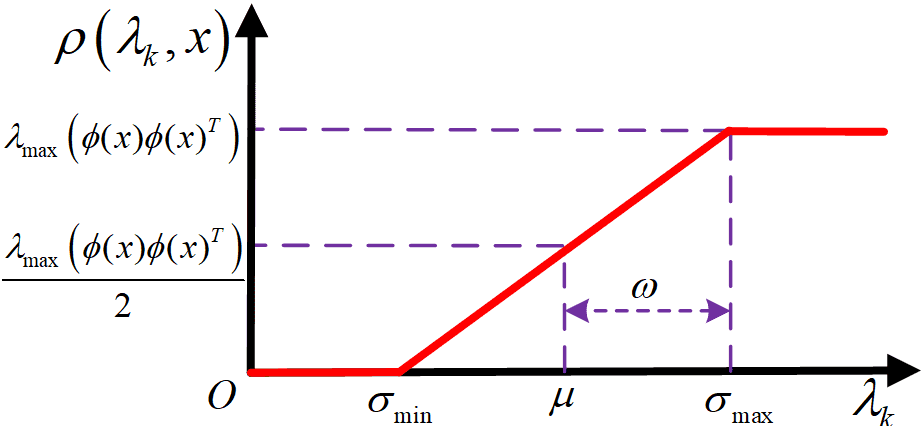}
		\vspace{-0.2cm}
		\caption{\textcolor{blue}{Illustrating figure for the forgetting factor}}
		\label{factor}
		\vspace{-0.2cm}
	\end{figure}
	
	According to the definitions of the forgetting factor in (\ref{forgetting}) and \textcolor{blue}{the excitation} matrix $W(t)$ in (\ref{W}), the \textcolor{blue}{previously collected} excitation information that belongs to the \textcolor{blue}{relatively weak excitation directions ($\lambda_k \leq \sigma_{\min}$) will not be attenuated, while that belongs to the strong excitation directions ($\lambda_k > \sigma_{\min}$)} will be forgotten with the rate $\beta(\lambda_k,x)$.

	\begin{remark}
		As discussed in the related results \cite{li2022concurrent,long2023filtering,parikh2019integral}, the derivative of $x$ in $\varpi(\dot{x},x,u)$ and (\ref{Z}) may be unavailable in many practical applications.
		\textcolor{blue}{The existing results generally utilize stable low-pass filters to construct the transformed system model and further avoid the direct utilization of $\dot{x}$, as introduced in (\ref{Y}) and (\ref{Q}).}
		However, we have found that \textcolor{blue}{for the integrator system models (\ref{model})}, the vector $Z(t)$ can also be obtained via ``integration by substitution'' method in \cite{hass2024thomas}, as shown below,
		\textcolor{blue}{
		\begin{align}
			Z(t)
			=& \int_{0}^{t} \varpi(\dot{x}(\tau),x(\tau),u(\tau)) - \sum_{k=1}^{h(\tau)} \frac{\beta(\lambda_k,x)}{\lambda_k(\tau)} E_k(\tau) Z(\tau) d\tau  \nonumber \\ 
			=& \int_{x(0)}^{x(t)}  \phi(x) dx - \int_{0}^{t} \phi (x(\tau)) \big(f(x(\tau))+u(\tau)\big) d\tau  \nonumber \\
			&- \sum_{k=1}^{h(t)} \int_{0}^{t} \frac{\beta(\lambda_k,x)}{\lambda_k(\tau)} E_k(\tau) Z(\tau) d\tau. \label{remark_2} 
		\end{align} 
		By avoiding the utilization of the stable low-pass filters, the above calculation method (\ref{remark_2}) can reduce the computational cost to a certain extent.}
		The effectiveness of this method will be demonstrated in the simulation section.
	\end{remark}
	
	The main \textcolor{blue}{results} of the LRE are \textcolor{blue}{summarized} in Theorem 1, which will be employed in the \textcolor{blue}{subsequent} design and analysis. 
	\begin{theorem}
		Considering the uncertain system (\ref{model}), with the proposed LRE (\ref{transfor_9}), the following results can be obtained.
		\begin{itemize}  
			\item [\textbf{(i)}] 
			\textcolor{blue}{The excitation matrix $W(t)$ remains positive semidefinite.}
			\item [\textbf{(ii)}] 
			If the regressor $\phi(x)$ is bounded, \textcolor{blue}{all the elements in $Z(t)$ and $W(t)$ remain} bounded.
			\item [\textbf{(iii)}] 
			For arbitrary moment $\tau \in [0,t]$, the range space of $\phi (x(\tau))$ is a subset of that of $W(t)$.
			\item [\textbf{(iv)}] 
			Given a constant vector $\varphi \in \mathbb{R}^p$, $W(t) \varphi =0$ is equivalent to $\phi (x(\tau))^T \varphi =0$ for all $\tau \in [0,t]$.
		\end{itemize}
	\end{theorem}

	\begin{proof}
		\textbf{(i)}
		\textcolor{blue}{For a specific moment $t$, if $W(t)$ is symmetric, its spectral decomposition (\ref{LRE_2}) exists and the spectral matrices are also symmetric.
		Then, it can be demonstrated from (\ref{W}) that} the derivative of $W(t)$ will be symmetric.
		Noting that the initial value $W(0)=O^{p \times p}$ is symmetric, it can be deduced by employing the recursive analysis method that  $W(t)$ remains symmetric for all $t \geq 0$, which further supports the utilization of spectral decomposition technique in (\ref{LRE_2}). 
		\textcolor{blue}{Moreover, the excitation matrix $W(t)$ remains positive semidefinite, since the term $\phi(x) \phi(x)^T$ in (\ref{W}) is positive semidefinite and the forgetting factor (\ref{forgetting}) equals zero for all $\lambda_k \in \left[0, \sigma_{\min}\right]$.}
		\vspace{0.1cm}
		\\
		\noindent
		\textbf{(ii)}
		\textcolor{blue}{Since the matrix $\phi(x) \phi(x)^T$ in (\ref{W}) is symmetric, it has $p(p+1)/2$ free elements.
		It can be seen from the inequality $h(t) \leq p \leq p(p+1)/2$ that $\phi(x) \phi(x)^T$ may not be linearly expressed by the spectral matrices $E_k(t)$, $k=1,2,...,h(t)$.
		Without loss of generality, let $\bar{\Phi}(t)$ denote the component of $\phi(x) \phi(x)^T$ that is linearly independent of the spectral matrices, we obtain the following decomposition,
		\begin{equation}
			\phi(x) \phi(x)^T = \sum_{k=1}^{h(t)} a_k(t) E_k(t) + \bar{\Phi}(t),
			\label{Phi_decomposition}
		\end{equation}
		where $a_k(t)$, $k=1,2,...,h(t)$ is a group of weight coefficients.
		Substituting (\ref{W}) and (\ref{Phi_decomposition}) into the derivative of (\ref{LRE_2}), we obtain the following equation,
		\begin{equation}
			\begin{aligned}
				\dot{W} 
				&= \sum_{k=1}^{h(t)} \dot{\lambda}_k E_k(t) + \sum_{k=1}^{h(t)} \lambda_k(t) \dot{E}_k \\
				&= \sum_{k=1}^{h(t)} a_k(t) E_k(t) - \sum_{k=1}^{h(t)} \beta(\lambda_k,x) E_k(t) + \bar{\Phi}(t).
			\end{aligned}
			\label{W_derivative}
		\end{equation}
		Since the spectral matrices $E_k(t)$, $k=1,2,...,h(t)$ and the component $\bar{\Phi}(t)$ are linearly independent, the following equations can be obtained for $k=1,2,...,h(t)$,
		\begin{equation}
			\dot{\lambda}_k E_k(t) = \big( a_k(t) - \beta(\lambda_k,x) \big) E_k(t). \label{W_derivative_1}
		\end{equation}
		It is shown from (\ref{forgetting}) and (\ref{Phi_decomposition}) that when $\lambda_k = \sigma_{\max}$, there is $\beta(\lambda_k,x) = \lambda_{\max}\left(\phi(x)\phi(x)^T\right) \geq a_k(t)$.
		Then the following inequality can be obtained from (\ref{W_derivative_1}),
		\begin{equation}
			\dot{\lambda}_k = a_k(t) - \beta(\lambda_k,x) \leq 0. \label{W_derivative_2}
		\end{equation}
		Thus the richness of the $k$th excitation direction should be reduced when $\lambda_k = \sigma_{\max}$, and the eigenvalues $\lambda_k$, $k=1,2,...,h(t)$ remain upper-bounded by $\sigma_{\max}$.
		Finally, it can be concluded from (\ref{LRE_2}) and (\ref{transfor_9}) that all the elements in $Z(t)$ and $W(t)$ remain bounded.}
		\vspace{0.1cm}
		\\
		\noindent
		\textbf{(iii)}
		According to the definitions of \textcolor{blue}{the excitation} matrix in (\ref{W}) and the forgetting factor in (\ref{forgetting}), if a column vector belongs to $\mathcal{R}[W(t)]$ at a specific moment $t$, it will always belong to $\mathcal{R}[W]$ after that.
		\textcolor{blue}{Thus $\text{rank}[W(t)]$ and $\dim [\mathcal{R}[W(t)]]$ are non-decreasing.}
		In other words, the spectra of previously collected excitation information will never be removed, and $\mathcal{R}[W(\tau)] \subseteq \mathcal{R}[W(t)]$ holds for all $\tau \in [0,t]$. 
		By utilizing the spectral matrices in (\ref{LRE_2}), the regressor $\phi (x(\tau))$ is decomposed onto the $h(t)$  \textcolor{blue}{excitation} directions, as shown below,
		\begin{equation}
			\phi (x(\tau)) = \sum_{k=1}^{h(t)} \phi_k (x(\tau)),  \label{decom_phi_1}
		\end{equation}
		\textcolor{blue}{where $\phi_k (x(\tau)) = E_k(\tau) \phi (x(\tau))$ denotes the $k$th component of the regressor.
		It is clear that $\mathcal{R}[\phi_k(x(\tau))] \subseteq \mathcal{R}[E_k(\tau)]$}, and $\mathcal{R}[\phi_k(x(\tau))]$, $k=1,2,...,h(t)$ are orthogonal to each other.
		From (\ref{W}), the relationship between $\mathcal{R}[\phi_k(x(\tau))]$ and $\mathcal{R}[W(t)]$ are discussed in the following two cases.\\
		$\bullet$ If $\lambda_k \leq  \sigma_{\min}$, the forgetting factor $\beta(\lambda_k,x)$ equals zero.
		\textcolor{blue}{Noticing that} $\dot{W}(\tau)$ contains the excitation term $\phi_k(x(\tau)) \phi(x(\tau))^T$ \textcolor{blue}{and it is not attenuated, the excitation information along the $k$th excitation direction is collected and} there is $\mathcal{R}[\phi_k(x(\tau))] \subseteq \textcolor{blue}{\mathcal{R}[W(\tau)] \subseteq} \mathcal{R}[W(t)]$.
		\\
		$\bullet$ If $\lambda_k > \sigma_{\min}$, the excitation information of the $k$th \textcolor{blue}{excitation} direction \textcolor{blue}{should have been sufficiently collected} during the interval $[0,\tau]$, then we further obtain $\mathcal{R}[\phi_k(x(\tau))] \subseteq \mathcal{R}[E_k(\tau)] \subseteq \mathcal{R}[W(\tau)] \subseteq \mathcal{R}[W(t)]$.
		
		According to the decomposition of $\phi(x(\tau))$ in (\ref{decom_phi_1}), $\mathcal{R}[\phi(x(\tau))]$ equals the sum of  $\mathcal{R}[\phi_k(x(\tau))]$, $k=1,2,...,h(t)$.
		Then it can be concluded from the above analysis that \textcolor{blue}{for all $\tau \in [0,t]$,} $\mathcal{R}\left[\phi(x(\tau))\right]$ is a subset of $\mathcal{R}\left[W(t)\right]$.
		\vspace{0.1cm}
		\\
		\noindent
		\textbf{(iv)}
		Without loss of generality, \textcolor{blue}{let us denote} the $i$th column vector of $W(t)$ and the $j$th column vector of $\phi(x(\tau))$ as $w_i(t)$ and $\phi_j(x(\tau))$, respectively.
		From Theorem 1 (iii), for arbitrary $j=1,2,...,n$, $\phi_j(x(\tau))$ can be linearly expressed by $w_i(t)$, $i=1,2,...,p$, as shown below,
		\begin{equation}
			\phi_j(x(\tau)) = \sum_{i=1}^{p} a_{ij}(t,\tau) w_i(t), \label{proof3_1}
		\end{equation}
		where $a_{ij}(t,\tau)$, $i=1,2,...,p$ is a group of weight coefficients.
		Moreover, we define a matrix $A(t,\tau) \in \mathbb{R}^{p \times n}$ as follows,
		\begin{equation}
			A(t,\tau)=
			\begin{bmatrix}
				a_{11} & a_{12} & \dots  & a_{1n} \\
				a_{21} & a_{22} & \dots  & a_{2n} \\
				\vdots & \vdots & \ddots & \vdots \\
				a_{p1} & a_{p2} & \dots  & a_{pn}
			\end{bmatrix},
		\end{equation}
		then the following equation holds,
		\begin{equation}
			\phi(x(\tau))=W(t) A(t,\tau). \label{proof3_2}
		\end{equation}
		Since $W(t)$ is a symmetric matrix, we have
		\begin{equation}
			\phi(x(\tau))^T \varphi = A(t,\tau)^T W(t)^T \varphi = A(t,\tau)^T W(t) \varphi. \label{proof3_3}
		\end{equation}
		It is clear that we can obtain $\phi(x(\tau))^T \varphi = 0$ from $W(t) \varphi=0$.
		Since $w_i(t)$ can be linearly expressed by all column vectors of $\phi(x(\tau))$ with $\tau \in [0,t]$, $W(t) \varphi=0$ can also be deduced from $\phi(x(\tau))^T \varphi = 0$.
		It can be concluded that $W(t) \varphi =0$ is equivalent to $\phi(x(\tau))^T \varphi =0$ for all $\tau \in [0,t]$.
	\end{proof}
	
	\begin{remark}
		Although some related results \cite{cho2018composite,pan2019efficient,lee2019concurrent,goel2020recursive,lai2024sift} have \textcolor{blue}{utilized forgetting factors to} guarantee the boundedness of $Z(t)$ and $W(t)$, \textcolor{blue}{the proposed scheme has unique advantages.
		Specifically, \cite{cho2018composite,pan2019efficient} directly add forgetting factors to both sides of the LRE to attenuate all the excitation directions, i.e., $\dot{Z}=-\beta Z+\varpi(\dot{x},x,u)$ and $\dot{W}=-\beta W + \phi(x)\phi(x)^T$, without considering the differences in excitation richness of the \textcolor{blue}{excitation} directions.
		Although \cite{lee2019concurrent,goel2020recursive,lai2024sift} present directional forgetting methods to attenuate the previously collected excitation information solely along the present excitation directions, it is still challenging to handle the weak excitation situations.}
		When the previously collected excitation information of the strong \textcolor{blue}{excitation} directions are forgotten, the weak \textcolor{blue}{excitation} directions may be unexpectedly attenuated at the same time, which may be harmful to the control and estimation performance.
		Compared with \textcolor{blue}{these methods, we utilize the spectral decomposition and the orthogonal projection techniques to decompose the previously collected excitation information based on excitation richness.
		By designing independent} forgetting factors $\beta(\lambda_k,x)$ for the independent excitation directions, \textcolor{blue}{the strong excitation directions are attenuated to guarantee the boundedness of $Z(t)$ and $W(t)$, while the excitation information along the weak excitation directions can be sufficiently collected.}
	\end{remark}
	
	\begin{remark}
		\textcolor{blue}{
		The main novelty of the proposed excitation collection scheme is ``decomposing the previously collected excitation information according to the excitation richness'', rather than ``designing the forgetting factor functions''.
		It is important to note that the main results of Theorem 1 can be obtained by utilizing many kinds of excitation attenuation schemes. 
		For example, a periodic excitation attenuation scheme is introduced as follows.}
		
		\textcolor{blue}{
		We denote $\Delta t \in \mathbb{R}^+$ as the attenuation period, then the sequence of attenuation moments are represented by $\mathcal{T}_i$, which satisfies $\mathcal{T}_0=0$ and $\mathcal{T}_i=\mathcal{T}_{i-1} + \Delta t$ for all $i \in \mathbb{Z}^+$. 
		Then the vector $Z(t)$ and the excitation matrix $W(t)$ are obtained by utilizing the following rules.} \\
		\textcolor{blue}{\textbf{1)} The initial values are set as $Z(0)=O^{p \times 1}$, $W(0)=O^{p \times p}$.} \\
		\textcolor{blue}{\textbf{2)} During the intervals between any two successive excitation attenuation moments, $Z(t)$ and $W(t)$ are obtained by integrating the transformed system model (\ref{transfor_1}), as shown below,
		\begin{equation}
			\begin{aligned}
				\dot{Z} &= \varpi(\dot{x},x,u), \text{ } \forall t \in \left( \mathcal{T}_{i-1}, \mathcal{T}_i \right), \\
				\dot{W} &= \phi(x) \phi(x)^T, \text{ } \forall t \in \left( \mathcal{T}_{i-1}, \mathcal{T}_i \right).
			\end{aligned}
			\label{LRE-descrete-1}
		\end{equation}
		\textbf{3)} At the excitation attenuation moments $\mathcal{T}_i$, $Z(\mathcal{T}_i^+)$ and $W(\mathcal{T}_i^+)$ are updated according to the decomposed LRE (\ref{direction}),
		\begin{equation}
			\begin{aligned}
				Z(\mathcal{T}_i^+) &= \sum_{k=1}^{h(\mathcal{T}_i^-)} \frac{\min \big( \lambda_k(\mathcal{T}_i^-),\eta \big)}{\lambda_k(\mathcal{T}_i^-)} E_k(\mathcal{T}_i^-) Z(\mathcal{T}_i^-),
				\\
				W(\mathcal{T}_i^+) &= \sum_{k=1}^{h(\mathcal{T}_i^-)} \min \big( \lambda_k(\mathcal{T}_i^-),\eta \big) E_k(\mathcal{T}_i^-),
			\end{aligned}
			\label{LRE-descrete-2}
		\end{equation}
		where $\eta \in \mathbb{R}^+$ is a design parameter.
		It is clear that the excitation richness of the strong excitation directions ($\lambda_k > \eta$) is reduced to $\eta$ at the attenuation moments, while the relatively weak excitation directions ($\lambda_k \leq \eta$) remain unchanged.}
	\end{remark}

	\subsection{Decomposition of Parameter Estimation Error}
	\textcolor{blue}{It can be seen from Theorem 1 (iii) that} all the spectra of historical excitation information during the interval $[0,t]$ are \textcolor{blue}{collected} into the \textcolor{blue}{excitation} matrix $W(t)$.
	The increase in ${\rm rank}[W(t)]$ suggests that some new excitation information has just been collected. 
	The regressor satisfies the \ref{SE} and \ref{IE} conditions \textcolor{blue}{if and only if ${\rm rank}[W(t)]=p$.}
	However, this paper focuses on \textcolor{blue}{adaptive control of uncertain nonlinear systems under non-persistent partial excitation, which indicates that the excitation information remains non-persistent and insufficient}, i.e., ${\rm rank}[W(t)]<p$, $\forall t \geq 0$.
	\textcolor{blue}{It is clear that for an adaptive control system, one of the SE / IE conditions and the non-persistent partial excitation conditions are satisfied.}
	
	\textcolor{blue}{Before presenting the decomposition method for the parameter estimation error, we first elaborate on its motivation.
	The main objective of adaptive control differs significantly from that of system identification. 
	System identification aims for estimation accuracy, while adaptive control aims to ensure satisfactory control performance. 
	That's to say, if a component of parameter estimation error has not affected the control performance, it can be tolerated in adaptive control design.}
	
	\textcolor{blue}{Noticing that the CE controllers generally utilize the estimation $\phi(x)^T \hat{\theta}$ to compensate for the parametric uncertainties $\phi(x)^T \theta$, the entire estimation error $\phi(x)^T \tilde{\theta}$ equals zero when $\tilde{\theta}$ belongs to the null space $\mathcal{N}\left[\phi(x)^T\right]$.
	It has been demonstrated form Theorem 1 (iv) that the excitation} matrix $W(t)$ shares the same null space with the regressor $\phi(x(\tau))$ for all $\tau \in [0,t]$, indicating that the projection of $\tilde{\theta}$ on $\mathcal{N}[W(t)]$ has not affected the control performance.
	Therefore, the null space $\mathcal{N}[W(t)]$ \textcolor{blue}{is defined} as the unexcited subspace for $\tilde{\theta}$.
	From Lemma 2 and \textcolor{blue}{Theorem 1 (i)}, the range space $\mathcal{R}[W(t)]$ and the null space $\mathcal{N}[W(t)]$ are orthogonal complement, thus $\mathcal{R}[W(t)]$ \textcolor{blue}{is defined} as the excited subspace for $\tilde{\theta}$.
	\textcolor{blue}{For a specific moment $t$, the parameter estimation error $\tilde{\theta}(t)$ is decomposed into the excited component $\tilde{\theta}_e(t)$ and the unexcited component $\tilde{\theta}_u(t)$, which are the orthogonal projections of $\tilde{\theta}(t)$ on the excited subspace $\mathcal{R}[W(t)]$ and the unexcited subspace $\mathcal{N}[W(t)]$ respectively, as shown below,
	\begin{align}
		\tilde{\theta}_e(t)
		&= {\rm Proj} \big( \tilde{\theta}(t),\mathcal{R}[W(t)] \big),  \label{excited} \\
		\tilde{\theta}_u(t)
		&= {\rm Proj} \big( \tilde{\theta}(t),\mathcal{N}[W(t)] \big).  \label{unexcited}
	\end{align}
	For more convenient formulation of the control design and analysis, the excited and unexcited components of the unknown vector and the parameter estimate are defined below,
	\begin{align}
		\theta_e(t) &= {\rm Proj} \big( \theta,\mathcal{R}[W(t)] \big),  \label{excited-theta} \\
		\theta_u(t) &= {\rm Proj} \big( \theta,\mathcal{N}[W(t)] \big),  \label{unexcited-theta} \\
		\hat{\theta}_e(t) &= {\rm Proj} \big( \hat{\theta}(t),\mathcal{R}[W(t)] \big),  \label{excited-hatheta} \\
		\hat{\theta}_u(t) &= {\rm Proj} \big( \hat{\theta}(t),\mathcal{N}[W(t)] \big).  \label{unexcited-hatheta}
	\end{align}
	Note that when the LRE (\ref{transfor_9}) collects some new excitation information, the excited and the unexcited components in (\ref{excited})-(\ref{unexcited-hatheta}) may experience abrupt changes as $\mathcal{R}[W(t)]$ and $\mathcal{N}[W(t)]$ vary suddenly.
	It is evident that the components $\tilde{\theta}_e$, $\tilde{\theta}_u$, $\theta_e$, $\theta_u$ are not accessible in practical applications.
	Consequently, they will be employed solely in stability analysis rather than in adaptive control design.}
	
	\begin{remark}
		Compared with the existing results on concurrent learning \cite{chowdhary2010concurrent,chowdhary2013concurrent,chowdhary2014exponential,li2022concurrent,long2023filtering,parikh2019integral} and composite learning \cite{pan2016composite,pan2018composite,roy2018combined,pan2017composite,guo2019composite,guo2020composite,cho2018composite,pan2019efficient,pan2022bioinspired}, the proposal of excited and unexcited subspaces is a significant improvement of this paper, with which the excited and the unexcited components of the parameter estimation error are promising to be discriminated in adaptive control design and stability analysis.
		\textcolor{blue}{It is important to} note that the \textcolor{blue}{decomposition method of parameter estimation error is suitable for not only} the proposed LRE (\ref{transfor_9}), \textcolor{blue}{but also} all the LREs that could collect all the spectra of \textcolor{blue}{historical} excitation information.
	\end{remark}

	\subsection{Adaptive Control Design and Stability Analysis}
	\textcolor{blue}{Based on the composite learning and $\mu$-modification techniques, we design an adaptive control law as shown below,
	\begin{align}
		u &= -k_e e -f(x) - \phi(x)^T \hat{\theta} + \dot{x}_r,  \label{CE}  \\
		\dot{\hat{\theta}} &= \gamma \phi(x) e + \gamma k_{\theta} \left(Z(t) - W(t) \hat{\theta}\right) - \mu W ^{\bot}(t) \hat{\theta}, \label{adaptive-law}
	\end{align}
	where $k_e$, $k_{\theta}$, $\gamma$, $\mu \in \mathbb{R}^+$ are design parameters.
	The matrix $W ^{\bot}(t)$ is defined as presented below
	\begin{equation}
		W ^{\bot}(t) = \begin{cases}
			E_1(t), &\text{if rank}[W(t)] < p, \\
			O^{p \times p}, &\text{if rank}[W(t)] = p.
		\end{cases}
		\label{proj-unexcited}
	\end{equation}
	Note that when $W(t)$ does not have full rank, $W ^{\bot}(t)$ is the unit projection matrix for the unexcited subspace $\mathcal{N}[W(t)]$.}
	
	\textcolor{blue}{In the parameter update law (\ref{adaptive-law}), the first (Lyapunov-based) term compensates for the control performance along the direction of $\phi(x)$, the second (composite learning) term accelerates the convergence of parameter estimation by utilizing the collected excitation information, the third ($\mu$-modification) term enhances the robustness of the closed-loop systems.
	From (\ref{model}), (\ref{control-error}), (\ref{estimation-error}), (\ref{CE}), (\ref{adaptive-law}), the dynamics of the resulting closed-loop system is obtained as follows,
	\begin{align}
		\dot{e} &= -k_e e + \phi(x)^T \tilde{\theta}, \label{closed-e} \\
		\dot{\tilde{\theta}} &= -\gamma \phi(x) e - \gamma k_{\theta} W(t) \tilde{\theta} + \mu W ^{\bot}(t) \hat{\theta}. \label{closed-theta}
	\end{align}
	It can be observed from (\ref{closed-theta}) that the composite learning term is} negative semidefinite with respect to $\tilde{\theta}$, which facilitates the enhancement of estimation and control performance. 
	\textcolor{blue}{Similar to the related results \cite{pan2016composite,pan2018composite,roy2018combined,pan2017composite,guo2019composite,guo2020composite,cho2018composite,pan2019efficient,pan2022bioinspired}, when the \ref{SE} and \ref{IE} conditions are satisfied, it can be readily shown, based on the positive definiteness of the excitation matrix $W(t)$, that the closed-loop system is exponentially stable. 
	However, this paper focuses on adaptive control of uncertain nonlinear systems under non-persistent partial excitation, in the sense that $W(t)$ is always positive semidefinite.}
	
	\textcolor{blue}{For a specific interval $[0,t]$, it can be shown from Theorem 1 (iii) that $\phi(x(\tau)) \in \mathcal{R}[W(t)]$ for all $\tau \in [0,t]$, which indicates that $-\gamma \phi(x(\tau)) e - \gamma k_{\theta} W(\tau) \tilde{\theta} \in \mathcal{R}[W(t)]$ and $\mu W ^{\bot}(\tau) \hat{\theta} \in \mathcal{N}[W(t)]$ always hold.
	From the definitions of $\tilde{\theta}_e$ in (\ref{excited}), $\hat{\theta}_u$ in (\ref{unexcited-hatheta}), and $W ^{\bot}(t)$ in (\ref{proj-unexcited}), we have $W(t) \tilde{\theta}(t) = W(t) \tilde{\theta}_e(t)$ and $W ^{\bot}(t) \hat{\theta}(t) = \hat{\theta}_u(t)$.
	Based on the above analysis, the dynamics of the parameter estimation error in (\ref{closed-theta}) can be further decoupled into the excited and the unexcited subspaces, as shown below,}
	\begin{align}
		\textcolor{blue}{\dot{\tilde{\theta}}_e} 
		& \textcolor{blue}{= -\gamma \phi(x) e - \gamma k_{\theta} W(t) \tilde{\theta}_e,} \label{closed-theta-e} \\
		\textcolor{blue}{\dot{\tilde{\theta}}_u}
		& \textcolor{blue}{= \mu \hat{\theta}_u.} \label{closed-theta-u}
	\end{align}
	
	%The main results of \textcolor{blue}{the proposed composite learning adaptive control scheme} are summarized in Theorem 2.
	\begin{theorem}
		Considering the uncertain system (\ref{model}) \textcolor{blue}{and a reference signal $x_r(t)$ under non-persistent partial excitation and Assumption 1}, by employing the composite learning adaptive control law (\ref{CE}), (\ref{adaptive-law}), the control error and the excited parameter estimation error component converge exponentially to zero, while the unexcited parameter estimation error component remains bounded.
	\end{theorem}

	\begin{proof}
		From Theorem 1 (iii) and the definition of $W(t)$ in (\ref{W}), ${\rm rank}[W(t)]$ increases monotonically in a discrete-time manner.
		From the monotone bounded theorem \cite{hass2024thomas} and the fact of ${\rm rank}[W(t)] < p$, ${\rm rank}[W(t)]$ must tend to its supremum as $t \rightarrow \infty$.
		Without out loss of generality, \textcolor{blue}{we denote} $\delta_{\kappa} \in \mathbb{Z}^+$ as the supremum of ${\rm rank}[W(t)]$, and let $t_{\kappa}$ be the first moment satisfying ${\rm rank}[W(t)]=\delta_{\kappa}$.
		\textcolor{blue}{Then the increasing moments of ${\rm rank}[W(t)]$ are defined as follows,}
		\begin{equation}
			t_i=\min \left\{ t>t_{i-1}: {\rm rank} [W(t)] > {\rm rank} [W(t_{i-1})] \right\},
		\end{equation}
		where $i=1,2,...,\kappa$ represent the serial numbers of these moments.
		We impose $t_0=0$ and $t_{\kappa+1}=\infty$ for the completeness of the definition.
		The proof \textcolor{blue}{will be completed} in three steps, which correspond to,
		(i) $e$ and $\tilde{\theta}_1$ during the intervals \textcolor{blue}{$\left(t_i,t_{i+1}\right)$},
		(ii) $e$ and $\tilde{\theta}_1$ during the entire time frame $[0,\infty)$, 
		\textcolor{blue}{and (iii) $\tilde{\theta}_2$ during the entire time frame $[0,\infty)$.}
		\vspace{0.1cm}
		\\
		\noindent
		\textbf{(i)}
		\textcolor{blue}{It is clear that $\mathcal{R}[W(t)]$ and $\mathcal{N}[W(t)]$ remain invariant during the intervals $\left(t_i,t_{i+1}\right)$, thereby the excited and the unexcited components in (\ref{excited})-(\ref{unexcited-hatheta}) vary continuously.}
		We construct a novel Lyapunov function that incorporates the control error $e$ and the excited component $\tilde{\theta}_{e}$, rather than the entire parameter estimation error $\tilde{\theta}$,
		\begin{equation}
			V_e (x,\hat{\theta}_e) = \frac{1}{2}e^Te +\frac{1}{2 \gamma} \tilde{\theta}_e^T \tilde{\theta}_e. \label{Vi}
		\end{equation}
		\textcolor{blue}{It can be shown from Theorem 1 (iv) and the definitions of $\tilde{\theta}_e$, $\tilde{\theta}_u$ in (\ref{excited}), (\ref{unexcited}) that $\phi(x(t))^T \tilde{\theta}(t) = \phi(x(t))^T \tilde{\theta}_e(t)$.}
		Then, from (\ref{closed-e}) and (\ref{closed-theta-e}), the derivative of the Lyapunov function (\ref{Vi}) is calculated as follows,
		\begin{align}
			\dot{V}_e (x,\hat{\theta}_e)
			=& -k_e e^Te + e^T\phi(x)^T \tilde{\theta} - \tilde{\theta}_e^T \phi(x) e - k_{\theta} \tilde{\theta}_e^T W(t) \tilde{\theta}_e
			  \nonumber \\
			=& -k_e e^Te - k_{\theta} \tilde{\theta}_e^T W(t) \tilde{\theta}_e.
			\label{dotV_e}
		\end{align}
			
		From \textcolor{blue}{Theorem 1 (i) and the definition of $\tilde{\theta}_e$ in (\ref{excited})}, $W(t)$ is positive semidefinite and $\tilde{\theta}_{e}(t)$ does not contain the eigenvector of $\lambda_1=0$, i.e., ${\rm Proj} \big( \tilde{\theta}_{e}(t), \mathcal{E}(\lambda_1) \big)=0$.
		According to Lemma 3, we obtain the following inequality,
		\begin{equation}
			\tilde{\theta}_e^T W(t) \tilde{\theta}_e \geq \lambda^+_{\min} (W(t)) \tilde{\theta}_e^T \tilde{\theta}_e,  \label{lambda}
		\end{equation}
		where $\lambda^+_{\min} (W(t))$ denotes the smallest positive eigenvalue of $W(t)$.
		Substituting (\ref{lambda}) into (\ref{dotV_e}), we have
		\begin{equation}
				\dot{V}_e (x,\hat{\theta}_e) \leq  -c(t) V_e (x,\hat{\theta}_e), \text{ } \forall t \in \textcolor{blue}{\left(t_i,t_{i+1}\right)},
			\label{dotV_i_2}
		\end{equation}
		where $c(t)= \min\left( 2 k_e, 2 \gamma k_{\theta} \lambda^+_{\min} (W(t)) \right)$ is the convergence rate.
		Then it can be concluded that the Lyapunov function (\ref{Vi}) converges exponentially to zero during the intervals \textcolor{blue}{$\left(t_i,t_{i+1}\right)$}.
		%\vspace{0.1cm}
		\\
		\noindent
		\textbf{(ii)}
		Note that the smallest positive eigenvalue $\lambda^+_{\min} (W(t))$ and the convergence rate $c(t)$ jump near zero just after the moments $t_i$, as shown in Fig. \ref{remark-fig}.
		\textcolor{blue}{To render} the proposed composite learning adaptive control scheme more \textcolor{blue}{persuasive}, the exponential stability will be further demonstrated for the entire time frame $[0,\infty)$ with a constant convergence rate. 
		\begin{figure}[!htbp]
			\vspace{-0.2cm}
			\centering
			\includegraphics[width=7cm]{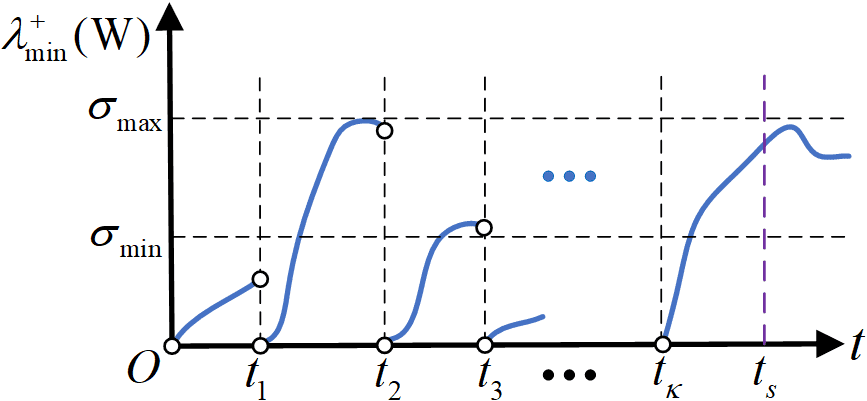}
			\vspace{-0.2cm}
			\caption{\textcolor{blue}{Illustrating figure for $\lambda^+_{\min} (W(t))$ and $t_s$.}}
			\label{remark-fig}
			\vspace{-0.2cm}
		\end{figure}
		
		\textcolor{blue}{The excited components of $\hat{\theta}(t)$ and $\tilde{\theta}(t)$ with respect to the excitation matrix at moment $t_{\kappa}$ are defined below,
		\begin{align}
			\hat{\theta}_{e,\kappa}(t) = {\rm Proj} \big( \hat{\theta}(t), \mathcal{R}[W(t_{\kappa}^+)] \big),
			\label{excited-entire-1} \\
			\tilde{\theta}_{e,\kappa}(t) = {\rm Proj} \big( \tilde{\theta}(t), \mathcal{R}[W(t_{\kappa}^+)] \big).
			\label{excited-entire-2}
		\end{align}
		To analyze the convergence property of $e$ and $\tilde{\theta}_{e,\kappa}(t)$ during the entire time frame, we define the following constant vectors,
		\begin{align}
			\varrho_i &= {\rm Proj} \big( \tilde{\theta}(0), \mathcal{R}[W(t_i^+)] \big), 
			\text{ } i=0,1,...,\kappa,  \label{vartheta_i} \\
			\varsigma_i &= {\rm Proj} \big( \tilde{\theta}(0), \mathcal{R}[W(t_{\kappa}^+)] \big) - \varrho_i.  \label{varsigma}
			\text{ } i=0,1,...,\kappa,   \\
			\vartheta_i &= \varrho_i - \varrho_{i-1} = \varsigma_{i-1}-\varsigma_i, \text{ } i=1,2,...,\kappa.  \label{varrho_i} 
		\end{align}
		Here, $\varrho_i$ denotes the component of $\tilde{\theta}_{e,\kappa}(0)$ that is excited during $[0,t_i]$, $\varsigma_i$ refers to the component that was not excited during $[0,t_i]$ but will be excited during $(t_i,t_{\kappa}]$, and $\vartheta_i$ represents the newly excited component at moment $t_i$.
		It is clear that $\tilde{\theta}_{e,\kappa}(0)$ and $\vartheta_i$ satisfy the following properties,
		\begin{align}
			\tilde{\theta}_{e,\kappa}(0) &= \sum_{i=1}^{\kappa} \vartheta_i, \label{decompose} \\
			\vartheta_{i_1}^T \vartheta_{i_2} &=0, \text{ } \forall i_1,i_2 = 1,2,...,\kappa \text{ and } i_1 \neq i_2,  \label{orthogonal} \\
			\left\Vert \varsigma_{i-1} \right\Vert &\geq \left\Vert \varsigma_i \right\Vert, \text{ } \forall i=1,2,...,\kappa.
			\label{inequality}
		\end{align}
		According to (\ref{decompose}) and Theorem 2 (i), the vector $[e(t),\tilde{\theta}_e(t)]$ converges to $[0,\varsigma_i]$ during the interval $(t_i,t_{i+1})$.
		It can be seen from (\ref{inequality}) that the vector $[e(t),\tilde{\theta}_{e,\kappa}(t)]$ converges monotonically to zero during the entire time frame.}
		\textcolor{blue}{Then we define a Lyapunov function for $e$ and $\tilde{\theta}_{e,\kappa}$, as shown below,
		\begin{equation}
			V_{e,\kappa}(x,\hat{\theta}_{e,\kappa}) = \frac{1}{2}e^Te +\frac{1}{2 \gamma} \tilde{\theta}_{e,\kappa}^T \tilde{\theta}_{e,\kappa}. \label{Ve}
		\end{equation}
		In what follows, $V_{e,\kappa}(x,\hat{\theta}_{e,\kappa})$ will be abbreviated as $V_{e,\kappa}(t)$.
		Similar to (\ref{dotV_i_2}), the derivative of (\ref{Ve}) satisfies $\dot{V}_{e,\kappa} (x,\hat{\theta}_{e,\kappa}) \leq  -c(t) V_{e,\kappa} (x,\hat{\theta}_{e,\kappa})$ during the interval $(t_{\kappa},\infty)$.}
		Without loss of generality, we arbitrarily choose a moment \textcolor{blue}{$t_s \in (t_{\kappa},\infty)$. 
		It can be observed from Fig. \ref{remark-fig} and the forgetting factor (\ref{forgetting}) that $\lambda_{\min}^+(W(t)) \geq \min \big(\sigma_{\min},\lambda_{\min}^+(W(t_s)) \big) >0$ holds for all $t \geq t_s$, and there should be $c(t) \geq \underline{c} \triangleq \min\left( 2 k_e, 2 \gamma k_{\theta} \sigma_{\min}, 2 \gamma k_{\theta} \lambda^+_{\min} (W(t_s)) \right)>0$.}
		\textcolor{blue}{For arbitrary $t > t_s$}, integrating (\ref{dotV_i_2}) \textcolor{blue}{over the interval $(t_s,t)$}, there is
		\begin{equation}
			V_{e,\kappa}(t) \leq V_{e,\kappa}(t_s)e^{-\underline{c}(t-t_s)}, \text{ } \forall t \geq t_s.  \label{bound_1}
		\end{equation}
		
		Since $V_{e,\kappa}(t)$ is decreasing during the interval $[0,t_s]$, we can conservatively obtain the following inequality,
		\begin{equation}
			\begin{aligned}
				V_{e,\kappa}(t) 
				&\leq V_{e,\kappa}(0)e^{-\underline{c}(t-t_s)}  \\
				&= V_{e,\kappa}(0)e^{\underline{c}t_s} e^{-\underline{c}t} \\
				&= \chi V_{e,\kappa}(0) e^{-\underline{c} t}, \text{ } \forall t \geq 0,
			\end{aligned}
			\label{bound_2}
		\end{equation}
		where $\chi=e^{\underline{c}t_s}$ is a constant corresponding to the \textcolor{blue}{excitation richness and the design parameters}.
		Fig. \ref{proof-fig} is given to show the relationship between the Lyapunov function $V_{e,\kappa}(t)$ and the exponential convergence bounds in (\ref{bound_1}), (\ref{bound_2}).
		\begin{figure}[!htbp]
			\vspace{-0.2cm}
			\centering
			\includegraphics[width=7cm]{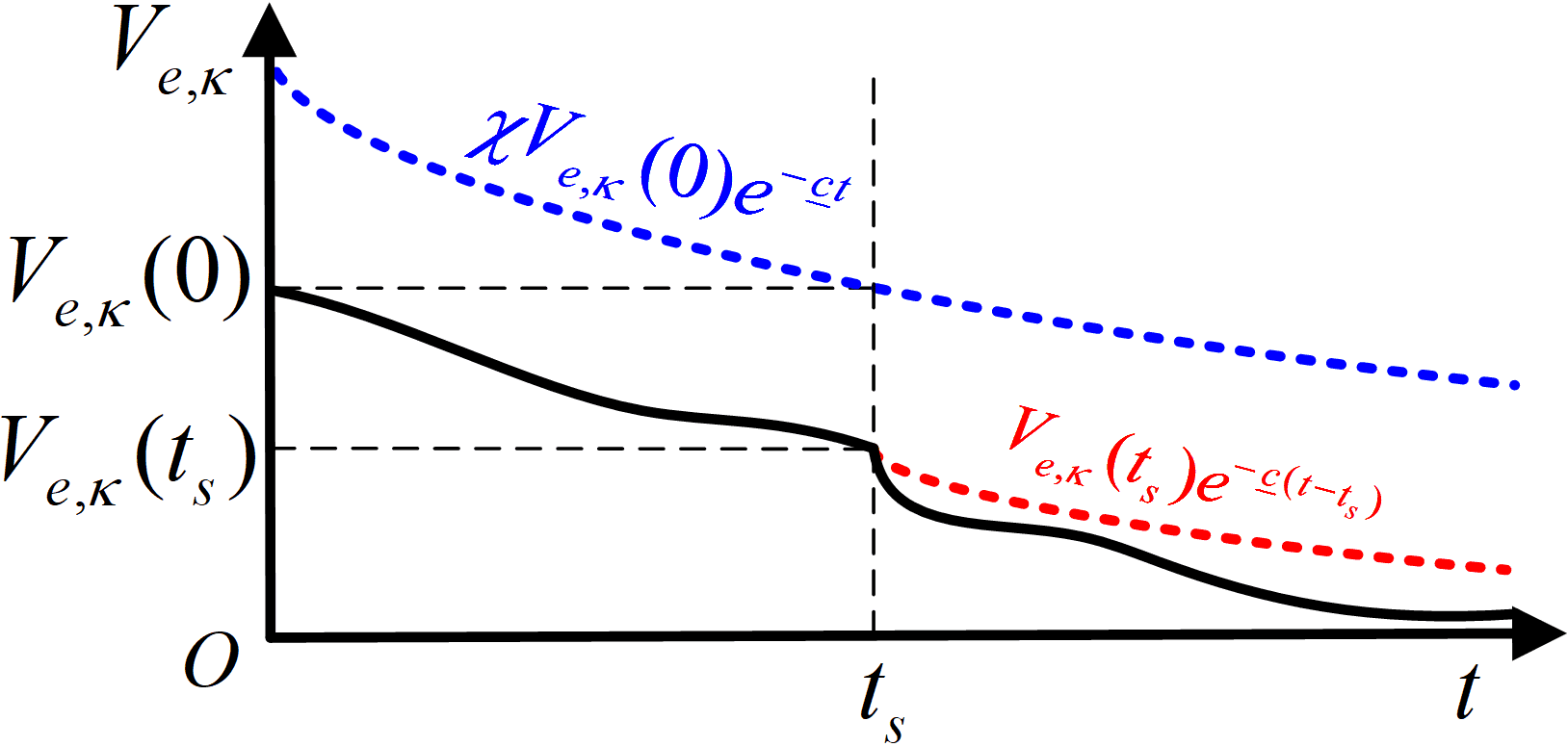}
			\vspace{-0.2cm}
			\caption{\textcolor{blue}{Illustrating figure for $V_{e,\kappa}(t)$ and the convergence bounds.}}
			\label{proof-fig}
			\vspace{-0.2cm}
		\end{figure}
		
		From (\ref{bound_2}), it can be concluded that the vector $[e(t),\tilde{\theta}_{e,\kappa}(t)]$ converge exponentially to zero during the entire time frame $[0,\infty)$, with a constant convergence rate $\underline{c}$.
		\vspace{0.1cm}
		\\
		\noindent
		\textbf{(iii)}
		\textcolor{blue}{Similar to the existing methods \cite{ioannou1983adaptive,narendra1987a,uzeda2023adaptive,uzeda2023robust,uzeda2025estimation}, we design a Lyapunov function for $\tilde{\theta}_u$, as shown below,
		\begin{equation}
			V_u (\hat{\theta}_u) = \frac{1}{2\mu} \tilde{\theta}_u^T \tilde{\theta}_u. \label{V_u}
		\end{equation}
		From (\ref{closed-theta-u}) and the Young's inequality, the derivative of (\ref{V_u}) is calculated for the intervals $(t_i,t_{i+1})$, as shown below,
		\begin{equation}
			\begin{aligned}
				\dot{V}_u (\hat{\theta}_u) 
				&= \tilde{\theta}_u^T \big( \theta_u - \tilde{\theta}_u \big) \\
				&\leq  \left( \frac{1}{2} \tilde{\theta}_u^T \tilde{\theta}_u + \frac{1}{2} \theta_u^T \theta_u \right) - \tilde{\theta}_u^T \tilde{\theta}_u \\
				&= -\mu V_u (\hat{\theta}_u) + \frac{1}{2} \theta_u^T \theta_u.
			\end{aligned}
			\label{dot_V_u}
		\end{equation}
		Noticing that $\frac{1}{2} \theta_u^T \theta_u$ is a positive constant, it can be concluded from (\ref{V_u}) and (\ref{dot_V_u}) that the unexcited parameter estimation error component $\tilde{\theta}_u$ remains bounded.}
	\end{proof}
	
	\begin{remark}
		Although some concurrent learning \cite{chowdhary2010concurrent,chowdhary2013concurrent,chowdhary2014exponential,li2022concurrent,long2023filtering,parikh2019integral} and composite learning \cite{pan2016composite,pan2018composite,roy2018combined,pan2017composite,guo2019composite,guo2020composite,cho2018composite,pan2019efficient,pan2022bioinspired} schemes have been proposed \textcolor{blue}{to enhance the estimation performance for adaptive control systems, they still require \ref{SE} and \ref{IE} conditions to ensure} exponential stability for system (\ref{model}).
		\textcolor{blue}{It is clear that these conditions are easier to be satisfied than \ref{PE} conditions, and can be verified online by checking $\text{rank}[W(t)]$.
		However, noticing that the \ref{SE} and \ref{IE} conditions rely on the future system states, i.e., $x(t)$ for $t \in [0,t_e]$, they are still challenging to ensure and verify in advance.
		In fact, the \ref{PE} conditions can be regarded as the constraints on both space and time dimensions, in the sense that the spectra of excitation information are sufficiently rich during the entire time frame.
		The existing results of concurrent learning and composite learning utilize the historical excitation information to relax the time constraint, thus the \ref{SE} and \ref{IE} conditions require the spectra of excitation information to be sufficiently rich just for an interval $[0,t_e]$.}
		
		\textcolor{blue}{Noticing that the main objective of adaptive control is to ensure control performance rather than estimation accuracy, the unexcited parameter estimation error component can generally be tolerated. 
		Based on the partial identification thoughts \cite{bittanti1990recursive,bittanti1992effective,kogan1996locally,li1998geometric}, this paper presents a novel composite learning adaptive control scheme to relax the excitation conditions.
		It is shown from (\ref{closed-e})-(\ref{closed-theta-u}) that under non-persistent partial excitation, the closed-loop system can be decoupled into the excited and the unexcited subsystems.
		The excited parameter estimation error component, as well as the effects of parametric uncertainties can be eliminated, while the unexcited component never affects the control performance.
		It is crucial to note that the proposed adaptive control scheme relaxes both space and time constraints, which indicates that the spectra of excitation information never need to be sufficiently rich.}
		Notably, this paper can be regarded as an extension of the existing results of composite learning. 
		The proposed adaptive controller \textcolor{blue}{provides} similar performance with the existing results under the \ref{SE} and \ref{IE} conditions, while \textcolor{blue}{it is also applicable to the non-persistent partial excitation situation.
		Noticing that for arbitrary moment $t$, one of the SE / IE condition ($\text{rank}[W(t)]=p$) and the non-persistent partial excitation condition ($\text{rank}[W(t)]<p$) is satisfied, the proposed adaptive control scheme can be implemented without imposing an excitation assumption.}
	\end{remark}
	
	\begin{remark}
		\textcolor{blue}{The existing findings of partial identification approaches \cite{bittanti1990recursive,bittanti1992effective,kogan1996locally,li1998geometric} generally decompose the parameter estimation error into the persistently excited and non-persistently excited components, and they eliminate the persistently excited component.
		Compared with these studies, the proposed adaptive control scheme sufficiently utilizes the historical excitation information, and eliminates the excited component.
		It is evident that the persistently excited component is always a part of the excited component.
		Consequently, the proposed adaptive control scheme demonstrates superior estimation and control performance. }

		\textcolor{blue}{In the existing results of $\mu$-modification technique  \cite{uzeda2023adaptive,uzeda2023robust,uzeda2025estimation}, the principal component analysis method is employed to estimate the persistently excited subspace under certain excitation assumptions.
		In comparison, this paper utilizes the excitation matrix to define the excited and unexcited subspaces, thereby relaxing the additional excitation assumptions and reducing the computational cost.}
	\end{remark}
	
	\begin{remark}
		\textcolor{blue}{Noticing that the LRE is orthogonal decomposed based on excitation richness in (\ref{direction}), it is promising to achieve some additional design objectives by slightly modifying the adaptive controller.
		For example, we can adjust the convergence rate for the different excitation directions by reconstructing a new LRE $Z_s(t)=W_s(t) \theta$ for the parameter update law (\ref{adaptive-law}),
		\begin{equation}
			\begin{aligned}
				Z_s(t) &= \sum_{k=1}^{h(t)} s_k(t) E_{k}(t) Z(t),  \\
				W_s(t) &= \sum_{k=1}^{h(t)} s_k(t) \lambda_{k}(t) E_{k}(t),
			\end{aligned}
			\label{reconstruction}
		\end{equation}
		where $s_k(t)$, $k=1,2,...,h(t)$ are design parameters regarding the convergence rate of the excitation directions.}
		
		\textcolor{blue}{On the other hand, due to the inevitable noises, the collected excitation information may become impure, then the relatively weak excitation directions may be overpowered and rendered useless for parameter estimation. 
		We can set the design parameters $s_k(t)$ of these excitation directions to be zero, and further avoid their unexpected effects.}
	\end{remark}
	
	\begin{remark}
		\textcolor{blue}{Let us concentrate on the implementation feasibility of the proposed adaptive control scheme.
		It is evident that the use of spectral decomposition may incur additional computational cost. 
		Nevertheless, it is essential to emphasize that the computational complexity of spectral decomposition is equivalent to that of matrix multiplication, namely $O(N_3)$, and can be resolved in real time by advanced microcomputers.}
		
		\textcolor{blue}{As revealed in the related results \cite{parikh2019integral}, the integrators in (\ref{Z}) and (\ref{W}) can effectively reduce the effects of other stochastic uncertainties for parameter estimation.
		Specifically, the effects of these stochastic uncertainties at different moments might cancel each other out during integration.}
		
		\textcolor{blue}{It can be observed from $\underline{c}$ that the convergence rate is related to the excitation richness and the design parameters $k_e$, $k_{\theta}$, $\gamma$, $\sigma_{\min}$.
		Here, appropriately increasing $k_e$, $k_{\theta}$, $\gamma$ can accelerate the convergence speed.
		Additionally, $\sigma_{\min}$ should be greater than the excitation richness required for accurate estimation, $\sigma_{\max}$ should not be overly large to prevent the issues of stiff differential equations and numerical divergence.}
	\end{remark}

	\section{Composite Learning Adaptive Dynamic Surface Control With RBFNN}
	By incorporating the \textcolor{blue}{radial basis function neural network (RBFNN) \cite{ge2002stable}} and dynamic surface control \cite{swaroop2002dynamic} techniques, the proposed composite learning adaptive control scheme is applied to the high-order systems with \textcolor{blue}{unstructured} uncertainties, which can be modeled as follows,
	\begin{align}
			\dot{x}_j &= \textcolor{blue}{x_{j+1} + f_j(\bar{x}_j) + \psi_j(\bar{x}_j) + d_j(t),} \text{ } \textcolor{blue}{j=1,2,...,n-1,} \nonumber \\
			\dot{x}_n &= \textcolor{blue}{u + f_n(x) + \psi_n(x) + d_n(t),} \label{high_order}
	\end{align}
	where $x=\left[x_1,x_2,...,x_n\right]^T \in \mathbb{R}^n$ and $u \in \mathbb{R}$ are the system states and control input, respectively.
	\textcolor{blue}{$\bar{x}_j=\left[x_1,x_2,...,x_j\right]^T \in \mathbb{R}^j$ denotes the first $j$ system states.
	$d_j(t)$, $j=1,2,...,n$ represent the unknown external disturbances.}
	$f_j: \mathbb{R}^j \rightarrow \mathbb{R}$ and $\psi_j: \mathbb{R}^j \rightarrow \mathbb{R}$ are the known and the unknown mappings, respectively.
	The control objective is to force the output $x_1$ to track a reference signal $x_r(t)$ \textcolor{blue}{under} the following assumption.
	\begin{assumption}
		The reference signal $x_r(t)$ and its first $n$ order derivatives, as well as the external disturbances $d_j(t)$, $j=1,2,...,n$, are piecewise continuous and bounded.
	\end{assumption}

	\textcolor{blue}{To obtain the parametric model for adaptive control design, we first utilize RBFNN to approximate the unknown mappings $\psi_j(\bar{x}_j)$ for $x \in \Omega$ with $\Omega$ of a compact set, as shown below,
	\begin{equation}
		\psi_j(\bar{x}_j) = \phi_j(\bar{x}_j)^T \theta_j^* + \delta_j(\bar{x}_j), \text{ } j=1,2,...,n, \label{RBFNN}
	\end{equation}
	where $\phi_j(\bar{x}_j)=[\phi_{j,1}(\bar{x}_j),...,\phi_{j,p}(\bar{x}_j)]^T$ is the $j$th regressor, with $\phi_{j,l}(\bar{x}_j) = \exp\left[-(\bar{x}_j-\upmu_{j,l})^T(\bar{x}_j-\upmu_{j,l})/\upeta_{j,l}^2\right]$, $l=1,2,...,p$.
	Here, $\upmu_{j,l} \in \mathbb{R}^p$ and $\upeta_{j,l} \in \mathbb{R}^+$ represent the center and width of the radius basis function $\phi_{j,l}(\bar{x}_j)$.
	$\theta_j^* \in \mathbb{R}^p$ denotes the optimal weights of the neural network, and $\delta_j(\bar{x}_j)$ is the approximation error.
	We denote $\Delta_j(\bar{x}_j,t)=\delta_j(\bar{x}_j)+d_j(t)$ as the non-parametric uncertainties of the $j$th subsystem.
	According to Assumption 2 and the universal approximation theorem \cite{ge2002stable}, $\Delta_j(\bar{x}_j,t)$ remains bounded for all $x \in \Omega$.
	For the sake of convenience, in what follows, $f_j(\bar{x}_j)$, $\psi_j(\bar{x}_j)$, $\phi_j(\bar{x}_j)$ and $\Delta_j(\bar{x}_j,t)$ are abbreviated as $f_j$, $\psi_j$, $\phi_j$ and $\Delta_j$, respectively.
	To construct a LRE} for the high-order uncertain system (\ref{high_order}), \textcolor{blue}{we first present its compact form as shown below,}
	\begin{equation}
		\dot{x} = F(x,u) + \Phi(x)^T \theta^* + \Delta(x,t), \label{compact}
	\end{equation}
	where the expressions of the mapping $F(x,u)$, the regressor $\Phi(x)$, \textcolor{blue}{the unknown optimal weight vector $\theta^*$ and the non-parametric uncertain term $\Delta(x,t)$ are presented as follows,}
	\begin{equation}
		\begin{aligned}
			F(x,u) &= \left[ x_2+f_1,...,x_n+f_{n-1}, u+f_n \right]^T, \\
			\Phi(x) &= \text{blockdiag}\left( \phi_1,\phi_2,...,\phi_n \right),   \\
			\textcolor{blue}{\theta^*} &= \textcolor{blue}{\left[ \theta_1^*,\theta_2^*,...,\theta_n^* \right]^T,}  \\
			\textcolor{blue}{\Delta(x,t)}&= \textcolor{blue}{\left[ \Delta_1,\Delta_2,...,\Delta_n \right]^T.}
		\end{aligned}
		\label{compactform}
	\end{equation}
	Then the transformed system model is given below,
	\begin{equation}
		\varpi_h(\dot{x},x,u) = \Phi(x) \Phi(x)^T\theta^* \textcolor{blue}{+ \Phi(x)\Delta(x,t),}
		\label{transfor_h1}
	\end{equation}
	where $\varpi_h(\dot{x},x,u)=\Phi(x) \left(\dot{x}-F(x,u) \right)$.
	\textcolor{blue}{By utilizing} the transformed system model (\ref{transfor_h1}) and the decomposed LRE (\ref{decomposition}), the LRE is designed as follows,
	\begin{equation}
		\mathcal{Z}(t) = \mathcal{W}(t) \theta^* \textcolor{blue}{+ \Delta_{\Phi}(t),} \label{LREh} 
	\end{equation}
	\textcolor{blue}{where $\mathcal{Z}(t)$ and $\mathcal{W}(t)$ are obtained by the following ODEs,}
	\begin{align}
		\dot{\mathcal{Z}} &= \varpi_h(\dot{x},x,u) - \sum_{k=1}^{h(t)} \textcolor{blue}{\frac{\beta(\lambda_k,x)}{\lambda_k(t)} E_k(t) \mathcal{Z}(t),} \label{Zh} \\
		\dot{\mathcal{W}} &= \Phi(x) \Phi(x)^T- \sum_{k=1}^{h(t)} \textcolor{blue}{\beta(\lambda_k,x) E_k(t).} \label{Wh} 
	\end{align}
	\textcolor{blue}{Here, $\Delta_{\Phi}(t)$ represents the bounded unknown vector caused by the external disturbances and the approximation errors.}

	\textcolor{blue}{We denote $\alpha_{j}$ as the $j$th virtual control signal and $\Phi_j(x) = \text{blockdiag}\left( O^{p \times 1},...,\phi_j(x),...,O^{p \times 1} \right)$ as the $j$th regressor.
	Considering that the main design thought is similar to Section \uppercase\expandafter{\romannumeral2} and the references \cite{krstic1995nonlinear,ge2002stable,swaroop2002dynamic}, the composite learning adaptive dynamics surface control law is summarized in Table. \ref{HIGH-CE}, followed by a brief stability analysis.
	Here, $k_j$, $k_{\theta}$, $\gamma$, $\mu$, $\iota \in \mathbb{R}^+$ are design parameters.}
	
	%\vspace{-0.25cm}
	\begin{table}[!htbp]
		\small
		\caption{\textcolor{blue}{Adaptive Control Design}}
		\begin{tabular}{p{8.5cm}}  
			\hline
			\hline
			\vspace{0cm}
			\textcolor{blue}{Control Error Variables:} \\
			\vspace{-0.4cm} \textcolor{blue}{
			\begin{equation}
				\begin{aligned}
					z_1 =& x_1 - x_r,   \\
					z_j =& x_j - \bar{\alpha}_{j-1} -x_r^{(j-1)}, \text{ } j=2,...,n.
				\end{aligned}
				\vspace{-0.2cm}
				\label{z_i}
			\end{equation}}  \\
			\hline
			\vspace{0cm}
			\textcolor{blue}{Stable Filters:} \\
			\vspace{-0.4cm} \textcolor{blue}{
			\begin{equation}
				\begin{aligned}
					\dot{\bar{\alpha}}_j &= \iota \left( \alpha_j-\bar{\alpha}_j \right) + z_j,  \\
					\tilde{\alpha}_j &= \bar{\alpha}_j - \alpha_j, \text{ } j=1,...,n-1.
				\end{aligned}
				\vspace{-0.2cm}
				\label{filter}
			\end{equation}}  \\
			\hline
			\vspace{0cm}
			\textcolor{blue}{Virtual Control Signals:} \\
			\vspace{-0.5cm} \textcolor{blue}{
			\begin{equation}
				\begin{aligned}
					\alpha_1 =& -k_1 z_1 - f_1 - \Phi_1^T \hat{\theta},  \\
					\alpha_j =&  -z_{j-1} -k_j z_j - f_j - \Phi_j^T \hat{\theta} - \iota \tilde{\alpha}_{j-1}, \text{ } j=2,...,n.
				\end{aligned}
				\vspace{-0.2cm}
				\label{CE_controller}
			\end{equation}}  \\
			\hline
			\vspace{0cm}
			\textcolor{blue}{Adaptive Control Law:} \\
			\vspace{-0.4cm} \textcolor{blue}{
			\begin{equation}
				\begin{aligned}
					u =& \alpha_n +x_r^{(n)}, \\
					\dot{\hat{\theta}} 
					=& \gamma \Phi(x)z + \gamma k_{\theta} \left( \mathcal{Z}(t) - \mathcal{W}(t) \hat{\theta} \right) - \mu \mathcal{W}^{\bot}(t) \hat{\theta}.
				\end{aligned}
				\vspace{-0.3cm}
				\label{adaptive-law-h}
			\end{equation}}  \\
			\hline
			\hline
		\end{tabular}
		\label{HIGH-CE}
	\end{table}

	\begin{theorem}
		Considering the high-order uncertain system (\ref{high_order}) and a reference signal $x_r(t)$ under \textcolor{blue}{non-persistent partial excitation and} Assumption 2, by employing the \textcolor{blue}{composite learning adaptive dynamic surface control law (\ref{adaptive-law-h}), the resulting closed-loop system is semi-globally stable.}
	\end{theorem}
	\begin{proof}
		\textcolor{blue}{We denote $\hat{\theta}$ and $\tilde{\theta}$ as the parameter estimate and the estimation error of $\theta^*$, respectively.
		Similar to Section \uppercase\expandafter{\romannumeral2}, the parameter estimation error is decomposed into the excited component $\tilde{\theta}_e$ and the unexcited component $\tilde{\theta}_u$.
		Then a Lyapunov function is designed as shown below,
		\begin{equation}
			V_h 
			= \frac{1}{2} \sum_{j=1}^{n} z_j^2 + \frac{1}{2} \sum_{j=1}^{n-1} \tilde{\alpha}_j^2
			+ \frac{1}{2 \gamma} \tilde{\theta}_e^T \tilde{\theta}_e + \frac{1}{2 \mu} \tilde{\theta}_u^T \tilde{\theta}_u.
			\label{V_h}
		\end{equation}
		By substituting the system model (\ref{high_order}), the RBFNN (\ref{RBFNN}), the LRE (\ref{LREh}) and the adaptive control law (\ref{z_i})-(\ref{adaptive-law-h}), the derivative of (\ref{V_h}) over the intervals $(t_i,t_{i+1})$ is calculated as follows,
		\begin{equation}
			\begin{aligned}
				\dot{V}_h 
				=& - \sum_{j=1}^{n} k_j z_j^2 - \sum_{j=1}^{n-1} \iota \tilde{\alpha}_j^2 
				- k_{\theta} \tilde{\theta}_e^T \mathcal{W}(t) \tilde{\theta}_e 
				- \tilde{\theta}_u^T \tilde{\theta}_u  \\
				& + \sum_{j=1}^{n} z_j \Delta_j
				- \sum_{j=1}^{n-1} \tilde{\alpha}_j \dot{\alpha}_j 
				+ k_{\theta} \tilde{\theta}^T \Delta_{\Phi} + \tilde{\theta}_u^T \theta_u. 
			\end{aligned}
			\label{dot-V-h-1}
		\end{equation}
		Then we obtain the following inequalities by utilizing the Young's inequality,
		\begin{equation}
			\begin{aligned}
				z_j \Delta_j
				&\leq \epsilon_1 z_j^2 + \varepsilon_1 \Delta_j^2, \text{ } j=1,...,n, \\
				- \tilde{\alpha}_j \dot{\alpha}_j
				&\leq \epsilon_2 \tilde{\alpha}_j^2 + \varepsilon_2 \dot{\alpha}_j^2, \text{ } j=1,...,n-1, \\
				\tilde{\theta}^T \Delta_{\Phi}
				&\leq \epsilon_3 \tilde{\theta}_e^T \tilde{\theta}_e + \epsilon_3 \tilde{\theta}_u^T \tilde{\theta}_u
				+ \varepsilon_3 \Delta_{\Phi}^T \Delta_{\Phi}, \\
				\tilde{\theta}_u^T \theta_u 
				&\leq \epsilon_4 \tilde{\theta}_u^T \tilde{\theta}_u + \varepsilon_4 \theta_u^T \theta_u,
			\end{aligned}
			\label{Youngs}
		\end{equation}
		where $\epsilon_{*}$, $\varepsilon_{*} \in \mathbb{R}^+$ are constant parameters satisfying $\epsilon_{*} \varepsilon_{*} = 1/4$ for $*=1,2,3,4$.
		Substituting (\ref{Youngs}) into (\ref{dot-V-h-1}), we obtain,
		\begin{equation}
			\begin{aligned}
				\dot{V}_h 
				\leq& - \sum_{j=1}^{n} \left(k_j-\epsilon_1 \right) z_j^2 
				- \sum_{j=1}^{n-1} \left(\iota-\epsilon_2 \right) \tilde{\alpha}_j^2  \\
				&- k_{\theta} \left( \min \left( \lambda_{\min}^+(\mathcal{W}), \sigma_{\min} \right) - \epsilon_3 \right) \tilde{\theta}_e^T \tilde{\theta}_e  \\
				&- \left( 1 - \epsilon_3 k_{\theta} - \epsilon_4 \right) \tilde{\theta}_u^T \tilde{\theta}_u + \Lambda(t)  \\
				=& -c_h(t) V_h + \Lambda(t), 
			\end{aligned}
			\label{dot-V-h-2}
		\end{equation}
		where the expressions of $c_h(t)$ and $\Lambda(t)$ are given as follows,
		\begin{equation}
			c_h(t) = 2 \min \left( 
			\begin{aligned}
				&k_j-\epsilon_1, \iota - \epsilon_2, \mu \left( 1 - \epsilon_3 k_{\theta} - \epsilon_4 \right), \\
				&\gamma k_{\theta} \left( \min \left( \lambda_{\min}^+(\mathcal{W}), \sigma_{\min} \right) - \epsilon_3 \right),
			\end{aligned}
			\right),
		\end{equation}
		\begin{equation}
			\Lambda(t)
			= \varepsilon_1 \sum_{j=1}^{n} \Delta_j^2 
			+ \varepsilon_2 \sum_{j=1}^{n-1} \dot{\alpha}_j^2 
			+ \varepsilon_3 k_{\theta} \Delta_{\Phi}^T \Delta_{\Phi}
			+ \varepsilon_4 \theta_u^T \theta_u.
		\end{equation}
		It is evident that there exist proper parameters $\epsilon_{*}$, $\varepsilon_{*}$, $*=1,2,3,4$ such that $c_h(t)$ remains positive during the intervals $(t_i,t_{i+1})$.
		When the closed-loop signals belong to a compact set $\Omega_h = \left\{ x,\tilde{\alpha}_j | x^Tx + \sum_{j=1}^{n-1} \tilde{\alpha}_j^2 \leq \varUpsilon  \right\}$, the redundant term $\Lambda(t)$ is upper-bounded by an unknown positive constant.
		It can be concluded from (\ref{V_h}) and (\ref{dot-V-h-2}) that all the closed-loop signals converge to a compact set during the intervals $(t_i,t_{i+1})$.
		Noticing that the jumping values of the Lyapunov function (\ref{V_h}) at $t_i$ are bounded, the closed-loop system is semi-globally stable during the entire time frame $[0,\infty)$.}
	\end{proof}
	
	\begin{remark}
		\textcolor{blue}{Compared with the existing results of adaptive dynamic surface control \cite{wang2005neural,peng2013adaptive}, this study employs composite learning technique to construct the negative definite term for the excited parameter estimation error component $\tilde{\theta}_e$, rather than adopting the $\sigma$-modification method and the Young's inequality.
		The positive term regarding $\theta_e^T \theta_e$ can be removed from the derivative of the Lyapunov function, thereby enhancing the estimation and control performance.
		In fact, by utilizing composite learning technique to construct the negative definite term of the excited parameter estimation error component, the control performance of many adaptive control systems can be enhanced.}
	\end{remark}
	
	\section{Simulation Results}
	\textcolor{blue}{To illustrate the effectiveness of the theoretical findings, we present comparative simulation results for three uncertain systems, which correspond to Theorems 1 - 3, respectively.
	In the simulations, the sampling time is set as $0.001\text{s}$, the threshold of effective excitation richness is set as $0.01$.}

	\subsection{Simulation Results of a Two-Link Planar Robot Arm}
	\textcolor{blue}{We consider a two-link planar robot arm \cite{pan2018composite}, which can be represented by the following model,
	\begin{equation}
		M(q) \ddot{q} + C(q,\dot{q}) \dot{q} + D \dot{q} + G(q) = \tau, \label{robot_arm}
	\end{equation}
	where $q=[q_1,q_2]^T$ and $\tau=[\tau_1,\tau_2]^T$ are the joint angular positions and control torques, respectively.
	The inertia matrix $M(q)$, the centripetal-Coriolis matrix $C(q,\dot{q})$, the friction matrix $D$ and the gravitational matrix $G(q)$ are given below,
		\begin{align}
			&M(q) = \left[ \begin{matrix}
				m_{11} & m_{12} \\
				m_{21} & m_{22}
			\end{matrix} \right], \text{ }
			D = \left[ \begin{matrix}
				k_{v1} & 0 \\
				0 & k_{v2}
			\end{matrix} \right], \nonumber \\
			&C(q,\dot{q}) =  \left[ \begin{matrix}
				-m_2 l_1 l_{c2} \dot{q}_2 \sin q_2 & -m_2 l_1 l_{c2} (\dot{q}_1+\dot{q}_2) \sin q_2 \\
				 m_2 l_1 l_{c2} \dot{q}_1 \sin q_2 & 0
			\end{matrix} \right], \nonumber \\
			&G(q) = g \left[ \begin{matrix}
				m_1 l_{c1} \cos q_1 + m_2 ( l_{c2} \cos(q_1+q_2) + l_1 \cos q_1 )\\
				m_2 l_{c2} \cos(q_1+q_2)
			\end{matrix} \right], \nonumber \\
			&m_{11} = m_1 l_{c1}^2 +\mathcal{I}_1 +\mathcal{I}_2+m_2(l_1^2+l_{c2}^2+2l_1l_{c2} \cos q_2), \nonumber \\
			&m_{12} = m_{21} = m_2 l_{c2}^2 + \mathcal{I}_2 + m_2 l_2 l_{c2} \cos q_2, \nonumber \\
			&m_{22} = m_2 l_{c2}^2 + \mathcal{I}_2.
		\end{align}
	Here, for the $j$th link, $m_j$ is the mass, $l_j$ is the length, $\mathcal{I}_j$ is the moment of inertia, $k_{vj}$ is the coefficient of viscous friction, $l_{cj}$ is the distance from the previous joint to the centroid of the $j$th link.
	We assume these model parameters to be unknown, and introduce an auxiliary function as shown below,
	\begin{equation}
		H(q,\dot{q},v,\dot{v}) = M(q) \dot{v} + C(q,\dot{q}) v + D \dot{q} + G(q), \label{robot_auxiliary}
	\end{equation}
	where $v$ is an auxiliary variable.
	The composite learning adaptive control law \cite{pan2018composite} is introduced in Table. \ref{ROBOT-CONTROL}.}
	\begin{table}[htbp]
		\small
		\caption{\textcolor{blue}{Adaptive Control Design for Robot Arm}}
		\begin{tabular}{p{8.5cm}}  
			\hline
			\hline
			\vspace{0cm}
			\textcolor{blue}{Control Error Variables:} \\
			\vspace{-0.4cm} \textcolor{blue}{
			\begin{equation}
				\begin{aligned}
					e =& q - q_r,   \\
					e_f =& \dot{e} + k_f e.
				\end{aligned}
				\vspace{-0.2cm}
				\label{robot-error}
			\end{equation}}  \\
			\hline
			\vspace{0cm}
			\textcolor{blue}{Linearly Parameterized Robot Model:} \\
			\vspace{-0.4cm} \textcolor{blue}{
			\begin{equation}
				H(q,\dot{q},v,\dot{v}) = \Phi (q,\dot{q},v,\dot{v})^T \theta.
				\vspace{-0.2cm}
				\label{LPRM}
			\end{equation}
			where $\Phi (q,\dot{q},v,\dot{v})$ and $\theta$ are provided below,
			\vspace{-0.2cm}
			\begin{equation}
				\begin{aligned}
					\Phi (q,\dot{q},v,\dot{v}) =& \left[ \begin{matrix}
						\dot{v}_1 & \dot{v}_2 & \phi_{13} & \phi_{14} & \phi_{15} & \dot{q}_1 & 0 \\
						0 & \dot{v}_1+\dot{v}_2 & \phi_{23} & 0 & \phi_{25} & 0 & \dot{q}_2
					\end{matrix} \right]^T, \\
					\theta =& [\theta_1,\theta_2,\theta_3,\theta_4,\theta_5,k_{v1},k_{v2}]^T, \\
				\end{aligned}
				\vspace{-0.1cm}
				\nonumber
			\end{equation}  
			\begin{equation}
				\begin{aligned}
					\phi_{13} =& (2 \dot{v}_1 + \dot{v}_2) \cos q_2 - (\dot{q}_2 v_1 +(\dot{q}_1+\dot{q}_2)v_2) \sin q_2, \\
					\phi_{14} =& g \cos q_1, \text{ }
					\phi_{15} = \phi_{25} = g \cos (q_1+q_2), \\
					\phi_{23} =& \dot{v}_1 \cos q_2 + \dot{q}_1 v_1 \sin q_2, \\
					\theta_1 =& \mathcal{I}_1 + m_1 l_{c1}^2 + m_2 l_1^2 + \mathcal{I}_2 + m_2 l_{c2}^2, \text{ }
					\theta_2 = m_2 l_{c2}^2 + \mathcal{I}_2, \\
					\theta_3 =& m_2 l_1 l_{c2}, \text{ }
					\theta_4 = m_1 l_{c1} + m_2 l_1, \text{ }
					\theta_5 = m_2 l_{c2}.
				\end{aligned}
				\vspace{-0.2cm}
				\nonumber
			\end{equation}} \\
			\hline
			\vspace{0cm}
			\textcolor{blue}{Linear Regression Euqation: $Z_r(t) = W_r(t) \theta$} \\
			\textcolor{blue}{Low-Pass Filters:} \\
			\vspace{-0.3cm} \textcolor{blue}{
			\begin{equation}
				\tau_f = \frac{a}{s+a} \tau, \text{ }
				\Phi_f = \frac{a}{s+a} \Phi(q,\dot{q},\ddot{q}).
				\vspace{-0.2cm}
				\label{LPFs}
			\end{equation}}  \\
			\hline
			\vspace{0cm}
			\textcolor{blue}{Composite Learning Adaptive Control Law:} \\
			\vspace{-0.5cm} \textcolor{blue}{
			\begin{equation}
				\hspace{-0.25cm}
				\begin{aligned}
					\tau =& -k_q e_f + \Phi(q,\dot{q},v,\dot{v})^T \hat{\theta},  \\
					\dot{\hat{\theta}} =& \gamma \Phi(q,\dot{q},v,\dot{v}) e_f + \gamma k_{\theta} \big(Z_r(t) - W_r(t) \hat{\theta}\big) - \mu W_r ^{\bot}(t) \hat{\theta},  \\
					v =& \dot{q}_d - k_f e.
				\end{aligned}
				\vspace{-0.2cm}
				\label{ROBOT-ADAPTIVE}
			\end{equation}}  \\
			\hline
			\hline
		\end{tabular}
		\label{ROBOT-CONTROL}
	\end{table}

	\textcolor{blue}{We consider three distinct LREs for comparison, the initial values are all set as $Z_r(0)=O^{p \times 1}$, $W_r(0)=O^{p \times p}$.
	The vector $Z_r(t)$ and excitation matrix $W_r(t)$ are obtained by utilizing the ODEs in the following cases.\\
	\textbf{Case A. I.} The LRE without forgetting factor \cite{pan2016composite}.
	\begin{equation}
		\dot{Z}_r = \Phi_f \tau_f, \text{ } 
		\dot{W}_r = \Phi_f \Phi_f^T.
	\end{equation}
	\textbf{Case A. II.} The LRE with a dependent forgetting factor \cite{cho2018composite}.
	\begin{equation}
		\dot{Z}_r = \Phi_f \tau_f - \beta Z, \text{ } 
		\dot{W}_r = \Phi_f \Phi_f^T - \beta W .
	\end{equation}
	\textbf{Case A. III.} The LRE with independent forgetting factors,
	\begin{equation}
		\begin{aligned}
			\dot{Z}_r &= \Phi_f \tau_f - \sum_{k=1}^{h(t)} \frac{\beta(\lambda_k,x)}{\lambda_k(t)} E_k(t) Z_r(t), \\
			\dot{W}_r &= \Phi_f \Phi_f^T- \sum_{k=1}^{h(t)} \beta(\lambda_k,x) E_k(t).
		\end{aligned}
	\end{equation}
	The model parameters are set to be $m_1 = 2 \text{kg}$, $m_2 = 1.4 \text{kg}$, $l_1 = l_2 = 0.8 \text{m}$, $l_{c1} = l_{c2} = 0.4 \text{m}$, $\mathcal{I}_1 = 0.5 \text{kg} \cdot \text{m}^2$, $\mathcal{I}_2 = 0.1 \text{kg} \cdot \text{m}^2$, $g = 9.8 \text{N}/\text{s}^2$, $k_{v1} = 0.8 \text{N} \cdot \text{m} \cdot \text{s}$, $k_{v2} = 1.6 \text{N} \cdot \text{m} \cdot \text{s}$.
	The design parameters are chose as $\sigma_{\min}=5$, $\sigma_{\max}=10$, $k_q = 10$, $k_f = 3$, $k_{\theta} = 10$, $a=5$, $\beta=5$, $\gamma = 5$, $\mu=5$.
	The reference signal of $q$ is set as follows,
	\begin{equation}
		q_r(t) = 
		\begin{cases}
			\frac{\pi}{2} [\sin(t), \cos(t)]^T, &\text{ if } t < 2\pi, \\
			\frac{\pi}{2} [\sin(2t), \cos(2t)]^T, &\text{ if } t \geq 2\pi.
		\end{cases}
	\end{equation}
	The simulation results for Cases A. I - III are presented in Figs. \ref{Simu_first_non} - \ref{Simu_first_independent}.
	The first six sub-figures illustrate the system states, the control errors and the control torques, respectively. 
	The seventh sub-figure shows the maximum and minimum eigenvalues of the excitation matrix, the last two sub-figures present the estimates of the uncertain term $H(q,\dot{q},\ddot{q})$ in (\ref{robot_auxiliary}).}

	\begin{figure}[!htbp]
		\vspace{-0.2cm}
		\centering
		\includegraphics[width=8.75cm]{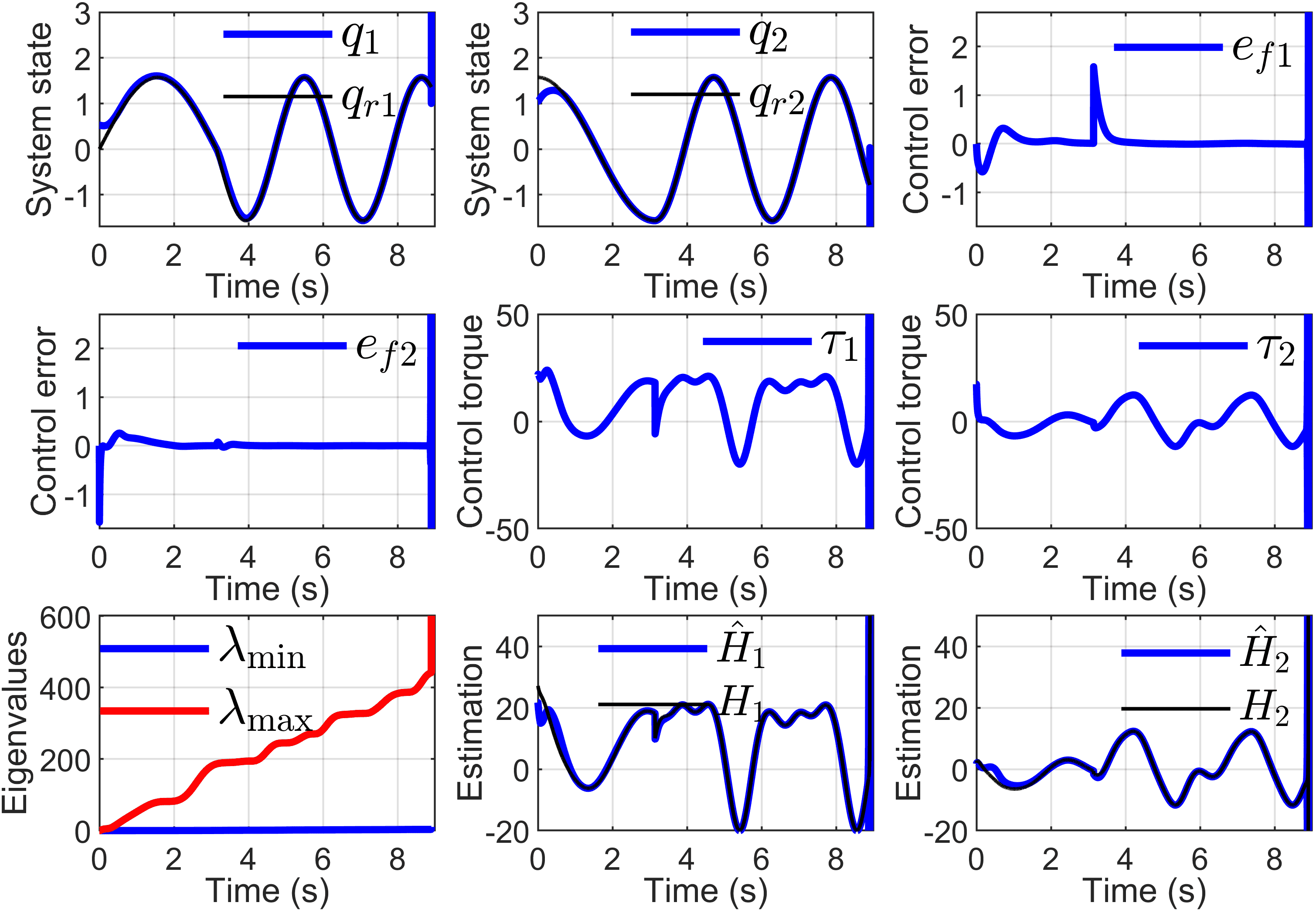}
		\vspace{-0.4cm}
		\caption{\textcolor{blue}{Simulation results in Case A. I.}} 
		\label{Simu_first_non}
		\vspace{-0.2cm}
	\end{figure}
	\begin{figure}[!htbp]
		\vspace{-0.2cm}
		\centering
		\includegraphics[width=8.75cm]{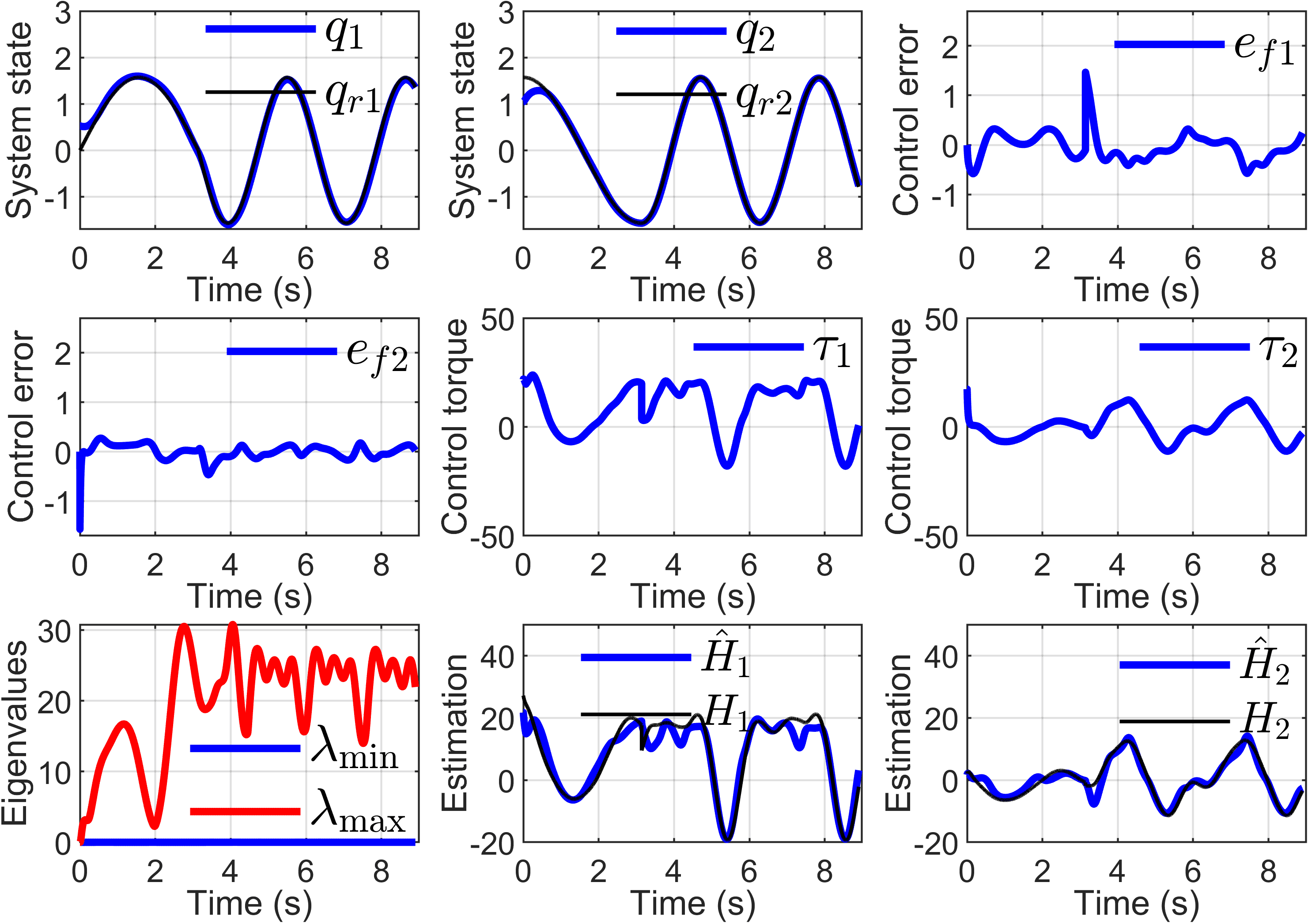}
		\vspace{-0.4cm}
		\caption{\textcolor{blue}{Simulation results in Case A. II.}}
		\label{Simu_first_dependent}
		\vspace{-0.2cm}
	\end{figure}
	\begin{figure}[!htbp]
		\vspace{-0.2cm}
		\centering
		\includegraphics[width=8.75cm]{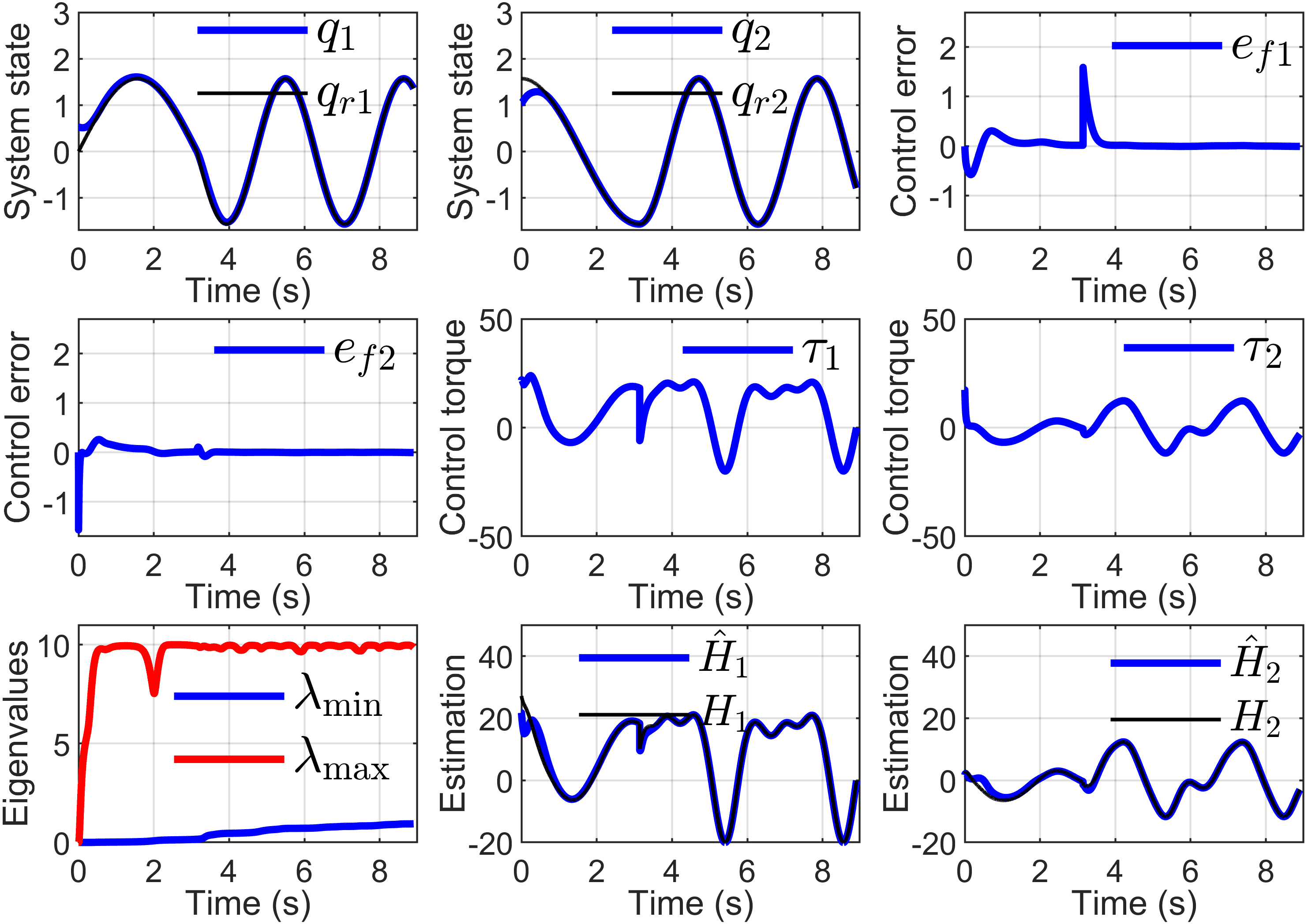}
		\vspace{-0.4cm}
		\caption{\textcolor{blue}{Simulation results in Case A. III.}} 
		\label{Simu_first_independent}
		%\vspace{-0.2cm}
	\end{figure}
		
	\textcolor{blue}{It can be observed from Fig. \ref{Simu_first_non} that when we utilize the LRE without forgetting factors, the eigenvalues of the excitation matrix increase without bound. This results in rapid parameter updates and numerical divergence at around $9 \text{s}$.}
	
	\textcolor{blue}{As shown in Fig. \ref{Simu_first_dependent}, when a dependent forgetting factor is applied to the entire LRE, although the boundedness of the eigenvalues is ensured, the weak excitation direction ($\lambda_{\min}$) is unexpectedly weakened. 
	This further affects the estimation and control performance.}
	
	\textcolor{blue}{It can be seen from Fig. \ref{Simu_first_independent} that when we employ the proposed LRE with independent forgetting factors, the excitation information along the weak excitation directions is sufficiently collected, and the eigenvalues of the excitation matrix remain bounded.
	The potential issue of numerical divergence is addressed, and the estimation and control performance are satisfactory.
	The comparative simulation results illustrate and verify the effectiveness of Theorem 1.}

	\subsection{Simulation Results of a First-Order Uncertain System under Sufficient and Non-Persistent Partial Excitation}
	Consider an uncertain nonlinear system (\ref{model}), \textcolor{blue}{with} $f(x)$, $\phi(x)$ and $\theta$ \textcolor{blue}{set} as follows,
	\begin{equation}
		\begin{aligned}
			f(x) &= \left[ x_1,x_2,\sin(x_3) \right]^T, \\
			\phi(x) &= \left[ \begin{matrix}
				x_2 & x_1 & \textcolor{blue}{x_2^2} \text{ }\\
				x_1 & x_2 & \textcolor{blue}{x_1^2} \text{ }\\
				\textcolor{blue}{\sin(x_2)} & \textcolor{blue}{\sin(x_1)} & x_3^2 \text{ }
			\end{matrix} \right], \\
			\theta &= \left[ \theta_1,\theta_2,\theta_3 \right]^T,
		\end{aligned}
		\label{simu_1}
	\end{equation}
	where $x=[x_1,x_2,x_3]^T$ and $u$ are the system states and control input, respectively.
	The  reference signal is set as $x_r(t) = [\sin(t),\sin(t),\cos(t)]^T$.
	
	The design parameters in the \textcolor{blue}{adaptive control law (\ref{transfor_9}), (\ref{forgetting}), (\ref{CE}), (\ref{adaptive-law})} are chosen as $\sigma_{\min}=5$, $\sigma_{\max}=10$ $k_e=5$, \textcolor{blue}{$k_{\theta}=50$, $\gamma=0.05$, $\mu=5$}.
	The unknown parameter \textcolor{blue}{vector} and the initial parameter estimates are set to be $\theta=[1,2,-1]^T$ and $\hat{\theta}(0)=[0.5,0.5,0.5]^T$, respectively.
	We consider two cases with different initial system states, as shown below. \\
	{\bf Case B. I.} $x(0)=[0.5,1.5,0.5]^T$, sufficient excitation.  \\
	{\bf Case B. II.} $x(0)=[2,2,0.5]^T$, insufficient excitation.
	
	\begin{figure}[htbp]
		\vspace{-0.2cm}
		\centering
		\includegraphics[width=8.75cm]{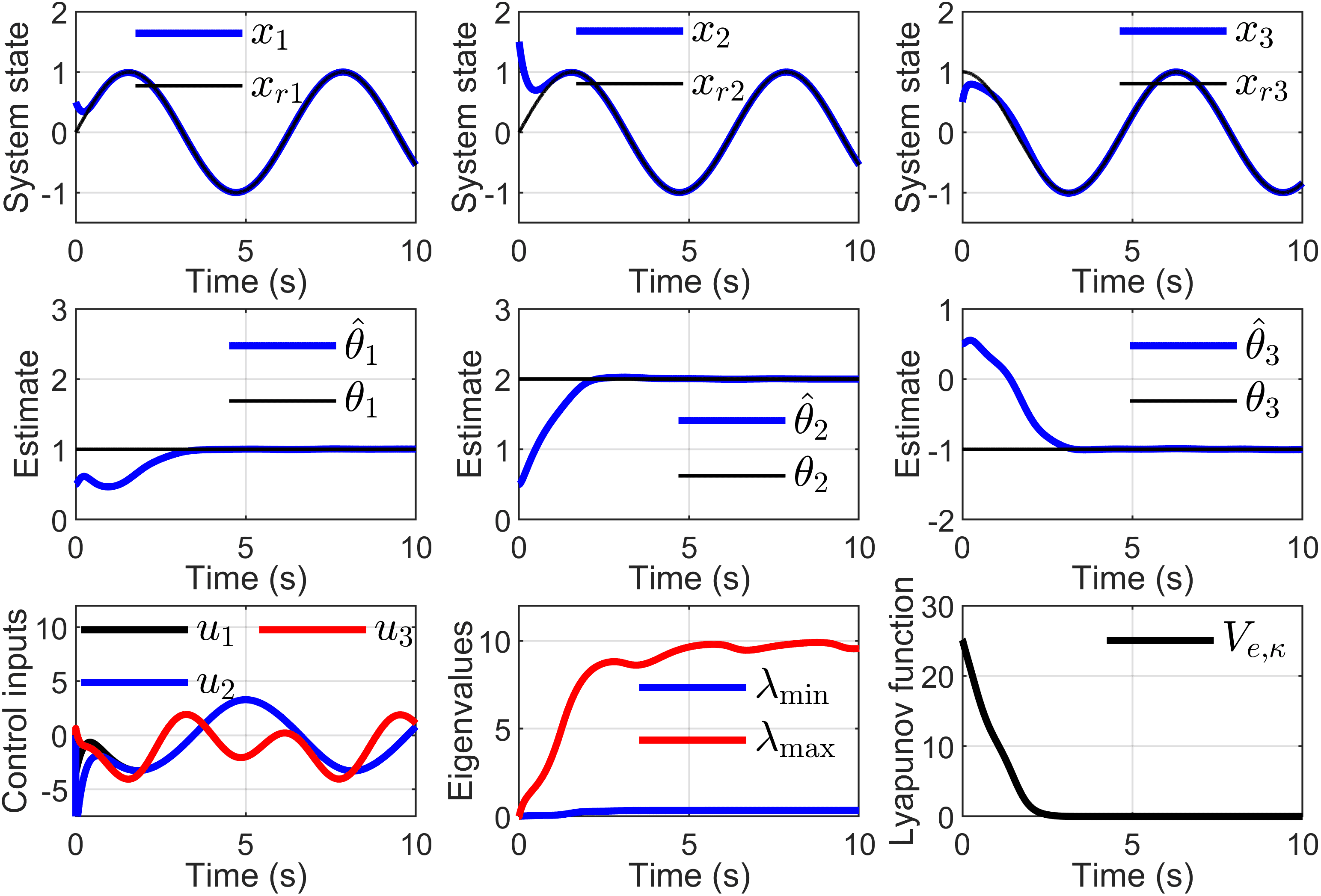}
		\vspace{-0.4cm}
		\caption{\textcolor{blue}{Simulation results in Case B. I.}} 
		\label{Simulation_1}
		%\vspace{-0.2cm}
	\end{figure}

	\begin{figure}[htbp]
		%\vspace{-0.2cm}
		\centering
		\includegraphics[width=8.75cm]{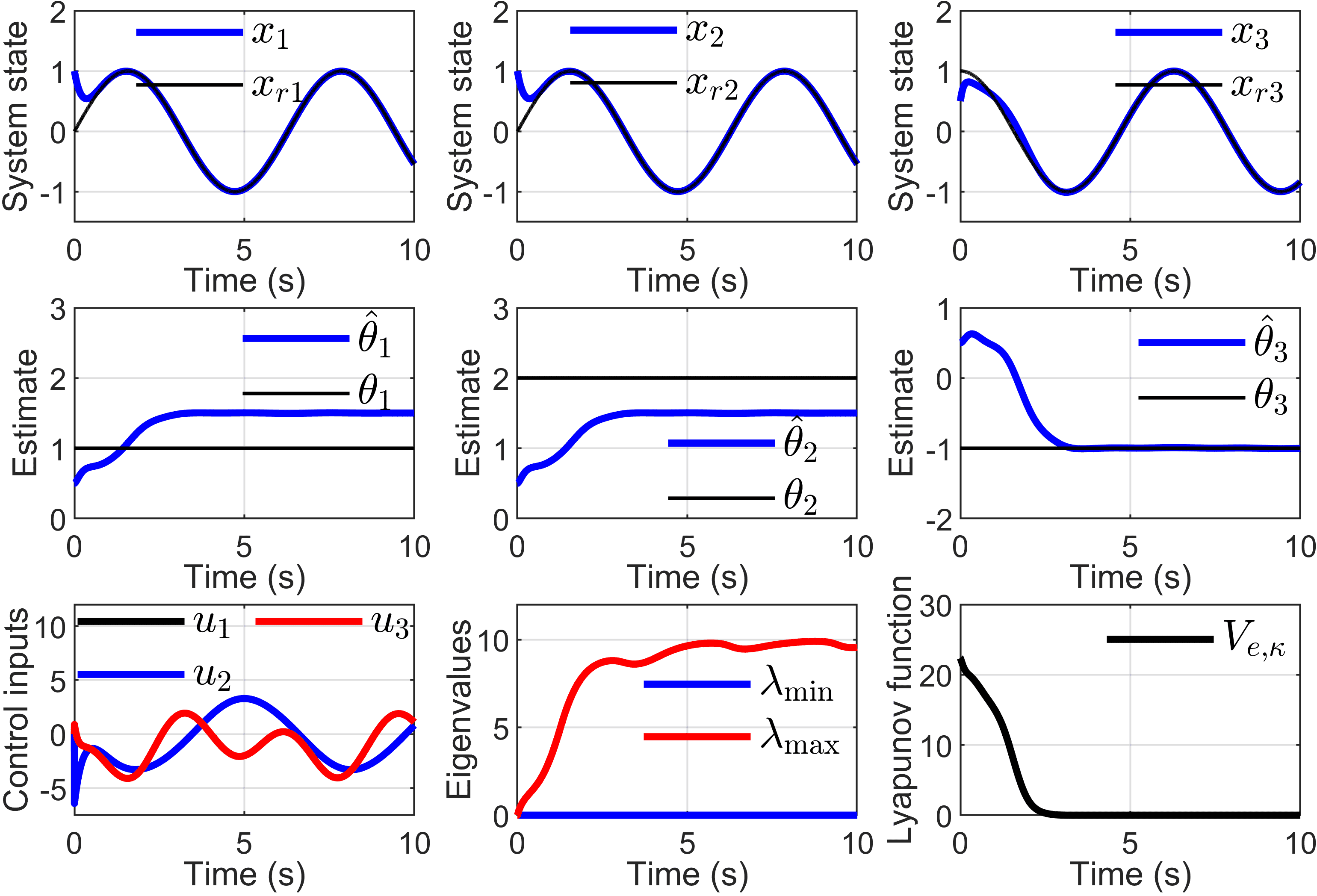}
		\vspace{-0.4cm}
		\caption{\textcolor{blue}{Simulation results in Case B. II.}}
		\label{Simulation_2}
		\vspace{-0.4cm}
	\end{figure}
	
	\textcolor{blue}{The simulation results for Case B. I and Case B. II are presented in Fig. \ref{Simulation_1} and Fig. \ref{Simulation_2}, respectively.
		The first seven sub-figures illustrate the system states, the parameter estimates and the control inputs, respectively. 
		The eighth sub-figure shows the maximum and minimum eigenvalues of the excitation matrix, the last sub-figure presents the Lyapunov function (\ref{Ve}).}
	
	\textcolor{blue}{In Case B. I, it can be verified from the regressor (\ref{simu_1}) and the initial system states that $\text{rank}[\phi(x(0))]=3$, which indicates that the \ref{SE} and \ref{IE} conditions are satisfied.
	It can be observed} from Fig. \ref{Simulation_1} that the tracking control errors and the parameter estimation errors converge to zero, the eigenvalues of the excitation matrix $W(t)$ remain upper-bounded by $\sigma_{\max}$.
	 
	\textcolor{blue}{In Case B. II, noticing that the initial values of $x_1$ and $x_2$ are equivalent, it can be deduced from} the system model (\ref{model}), (\ref{simu_1}) and \textcolor{blue}{the reference signal $x_r(t)$ that $x_1(t)$ and $x_2(t)$ will remain equivalent} for all $t \geq 0$, indicating that the first two columns of the regressor $\phi(x)$ are always equivalent, \textcolor{blue}{and the excitation information will be always insufficient.}
	From Fig. \ref{Simulation_2}, the system states and parameter estimates converge to zero and $[1.5,1.5,-1]^T$, respectively.
	To establish the relation between the theoretical findings and the simulation results, we present a theoretical analysis in the following.
	
	\textcolor{blue}{It can be seen from (\ref{W}) and (\ref{simu_1}) that when $x_1$ and $x_2$ are always equivalent, the term $\phi(x(t)) \phi(x(t))^T$ and the excitation matrix $W(t)$ satisfy the following form for all $t \geq 0$,
	\begin{equation}
		\left[ \begin{matrix}
			a & a & b\\
			a & a & b\\
			b & b & c
		\end{matrix} \right].
	\end{equation}
	Then} the range space and the null space of $W(t)$ \textcolor{blue}{can be readily} calculated as $\mathcal{R}[W(t)]={\rm span}\left\{ [1,1,0]^T,[0,0,1]^T \right\}$ and $\mathcal{N}[W(t)]={\rm span}\left\{ [1,-1,0]^T \right\}$, respectively.
	The \textcolor{blue}{excited and the unexcited components of the initial parameter estimate $\hat{\theta}(0)=[0.5,0.5,0.5]^T$ and the initial estimation error $\tilde{\theta}(0)=[0.5,1.5,-1.5]^T$ are calculated as follows,
	\begin{equation}
		\begin{aligned}
			\hat{\theta}_{e,\kappa}(0) &= {\rm Proj} \big( \hat{\theta}(0), \mathcal{R}[W] \big) = [0.5,0.5,0.5]^T, \\
			\hat{\theta}_{u,\kappa}(0) &= {\rm Proj} \big( \hat{\theta}(0), \mathcal{N}[W] \big) = [0,0,0]^T, \\
			\tilde{\theta}_{e,\kappa}(0) &= {\rm Proj} \big( \tilde{\theta}(0), \mathcal{R}[W] \big) = [1,1,-1.5]^T, \\
			\tilde{\theta}_{u,\kappa}(0) &= {\rm Proj} \big( \tilde{\theta}(0), \mathcal{N}[W] \big) = [-0.5,0.5,0]^T.
		\end{aligned}
	\end{equation}}
	
	According to the parameter update law (\ref{adaptive-law}) and Theorem 2, the excited component $\tilde{\theta}_{e,\kappa}$ will converge to zero, \textcolor{blue}{and the $\mu$-modification term $-\mu \mathcal{W}^{\bot}(t) \hat{\theta}$ remains zero.
	Then, it is clear that the parameter estimation error will converge to $\tilde{\theta}_{u,\kappa}(0)$, and the stead-state value of the parameter estimate will be
	\begin{equation}
		\hat{\theta}_{\infty}=\theta-\tilde{\theta}_{u,\kappa}(0)=[1.5,1.5,-1]^T. \label{simu_eq_1}
	\end{equation}
	The simulation results in Fig. \ref{Simulation_1} and Fig. \ref{Simulation_2} verify the above theoretical analysis and the effectiveness of the proposed composite learning adaptive control scheme.}

	\subsection{Simulation Results of a Third-Order System under Unstructured Uncertainties and External Disturbances}
	Consider a \textcolor{blue}{third-order uncertain} system, as shown below,
	\textcolor{blue}{
	\begin{equation}
		\begin{aligned}
			\dot{x}_1 &= x_2 + \psi_1(x_1) + d_1(t), \\
			\dot{x}_2 &= x_3 + \psi_2(x_2) + d_2(t), \\
			\dot{x}_3 &=  u  + \psi_3(x_3) + d_3(t),
		\end{aligned}
		\label{simu_model}
	\end{equation}
	where $x=[x_1,x_2,x_3]^T$ and $u$ are the system states and control input, respectively.
	The control objective is to force the output $x_1$ to track a reference signal, i.e., $x_r(t) = \sin(t)$.
	The unknown functions and external disturbances are set as 
	\begin{equation}
		\begin{aligned}
			\psi_1(x_1) &= 0.1x_1^2+0.5\sin(x_1), \text{ } d_1(t) = 0.05 \sin(2t), \\
			\psi_2(x_2) &= 0.1x_2^2+0.5\cos(x_2), \text{ } d_2(t) = 0.05 \cos(3t), \\
			\psi_3(x_3) &= 0.1x_3^2+0.5\sin(x_3), \text{ } d_3(t) = 0.05 \sin(4t).
		\end{aligned}
	\end{equation}
	Two distinct cases are considered in the simulation, corresponding to the standard adaptive dynamic surface control scheme and the proposed composite learning adaptive dynamic surface control scheme.}
	Specifically, we employ the same \textcolor{blue}{certainty equivalence} control law in (\ref{adaptive-law-h}) in these two cases, while two distinct parameter update laws are applied. \\
	{\bf Case C. I.} \textcolor{blue}{The ``Lyapunov-based'' parameter update law with $\sigma$-modification term, i.e., $\dot{\hat{\theta}} = \gamma \Phi(x)z - \sigma \hat{\theta}$.}  \\
	{\bf Case C. II.} The proposed parameter update law in (\ref{adaptive-law-h}).
	
	\textcolor{blue}{The design parameters in the RBFNN are chosen as $\upmu_{j,1}=-4$, $\upmu_{j,2}=-2$, $\upmu_{j,3}=0$, $\upmu_{j,4}=2$, $\upmu_{j,5}=4$, $\eta_j = 4$ for $j=1,2,3$.}
	The design parameters in the adaptive controllers are chosen as $\sigma_{\min}=5$, $\sigma_{\max}=10$, \textcolor{blue}{$k_1=1$, $k_2=1$, $k_3=1$, $k_{\theta}=20$, $\gamma=0.5$, $\mu=0.01$, $\sigma=0.01$, $\iota=8$.}
	The initial system states and initial parameter estimates are set as \textcolor{blue}{$x(0)=[0.1,0.9,0.1]^T$ and $\hat{\theta}(0)=O^{15 \times 1}$, respectively.
	The low-pass filters are initialized to be $\bar{\alpha}_1(0)=\alpha_1(0)$, $\bar{\alpha}_2(0)=\alpha_2(0)$.}
	
	\begin{figure}[!htbp]
		\vspace{-0.2cm}
		\centering
		\includegraphics[width=8.75cm]{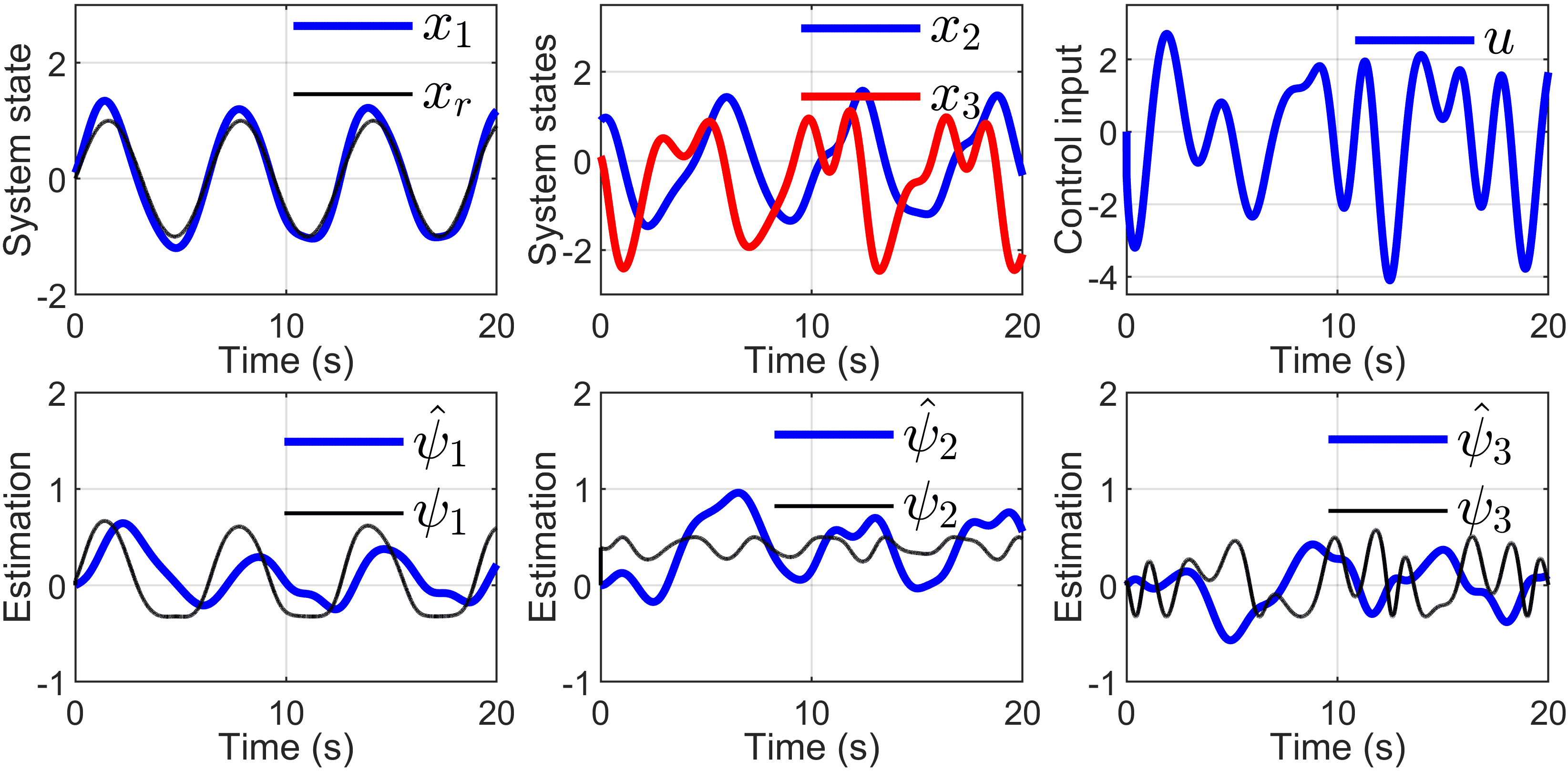}
		\vspace{-0.4cm}
		\caption{\textcolor{blue}{Simulation results in Case C. I.}}
		\label{Simulation_5}
		%\vspace{-0.2cm}
	\end{figure}
	
	\begin{figure}[!htbp]
		%\vspace{-0.2cm}
		\centering
		\includegraphics[width=8.75cm]{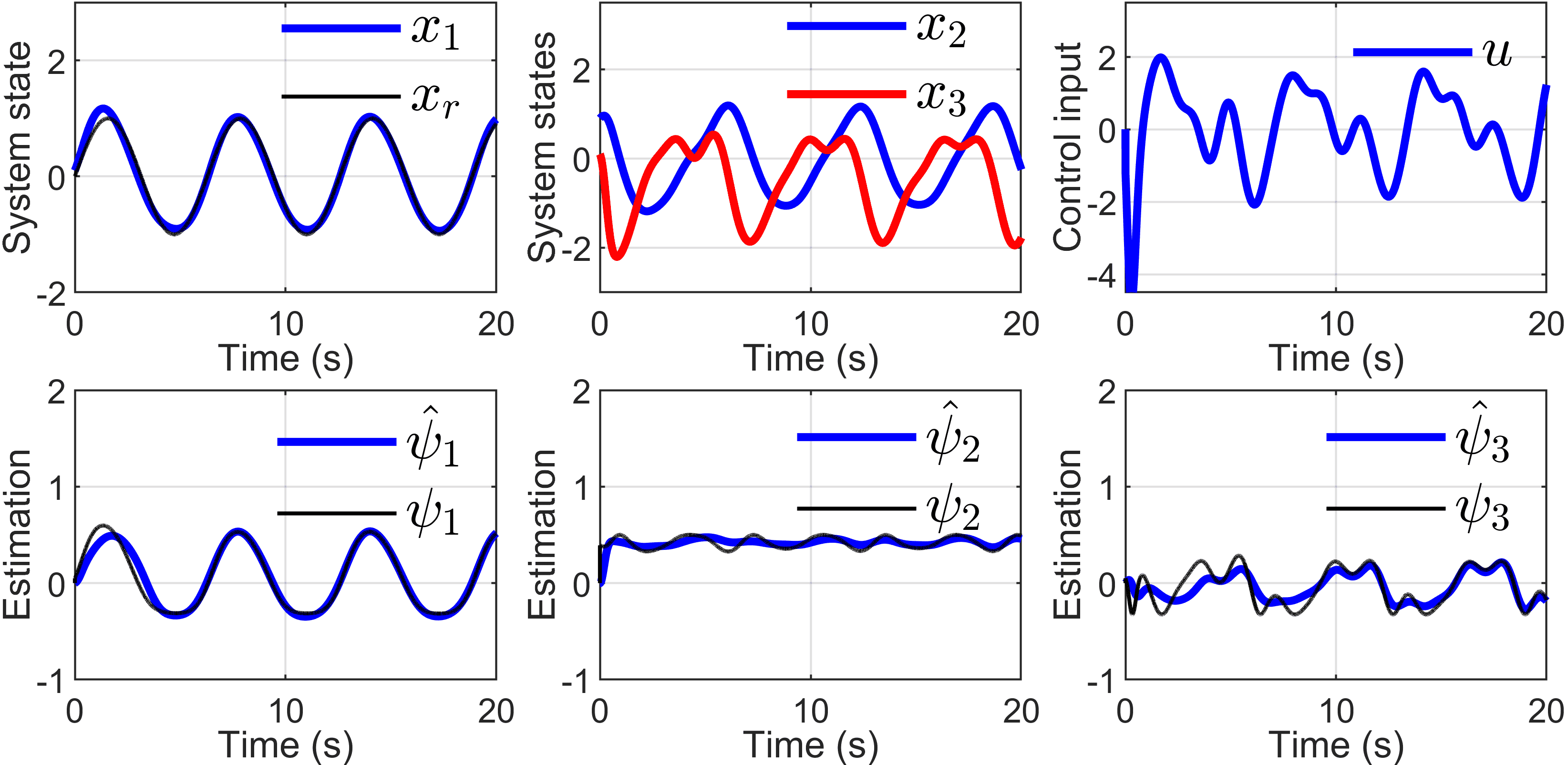}
		\vspace{-0.4cm}
		\caption{\textcolor{blue}{Simulation results in Case C. II.}} 
		\label{Simulation_6}
		\vspace{-0.5cm}
	\end{figure}
	
	Fig. \ref{Simulation_5} and Fig. \ref{Simulation_6} show the system states, \textcolor{blue}{the control inputs and uncertainty} estimates for Case C. I and Case C. II, respectively.
	\textcolor{blue}{Noticing that the RBFNN has $15$ weight parameters, we just provide the estimates for the unknown functions, i.e., $\hat{\psi}_j$, $j=1,2,3$. }
	
	In Case C. I, it can be seen from Fig. \ref{Simulation_5} that \textcolor{blue}{when we apply the standard dynamic surface control scheme \cite{wang2005neural,peng2013adaptive}, all the closed-loop signals remain bounded.
	However, it is clear that the RBFNN cannot accurately approximate the unknown functions, and there exists obvious tracking control error $z_1 = x_1-x_r$ at around $8\text{s}$ and $14\text{s}$.}
	
	In Case C. II, \textcolor{blue}{it can be observed from Fig. \ref{Simulation_6} that when we apply the proposed adaptive dynamic surface control scheme, all the closed-loop signals remain bounded.
	Noticing that the estimation errors and control errors in Case C. II are obviously smaller than that in Case C. I, the proposed scheme performs better than the existing results \cite{wang2005neural,peng2013adaptive}.
	The comparative simulation results illustrate and verify the effectiveness of Theorem 3 and the discussions in Remark 9.}

	\section{Conclusion}
	\textcolor{blue}{Based on the spectral decomposition, composite learning, and $\mu$-modification techniques}, we present a new adaptive control scheme for uncertain nonlinear systems.
	\textcolor{blue}{The linear regression equation sufficiently collects the historical excitation information and ensures the boundedness of the elements.
	The parameter estimation error is decomposed into the excited component and the unexcited component.
	It is demonstrated that the effects of parametric uncertainties are completely eliminated, and the robustness of the closed-loop systems is enhanced.
	Additionally, the proposed scheme is extended to composite learning adaptive dynamic surface control for high-order uncertain systems with unstructured uncertainties.}
	Future research efforts will focus on considering unknown control coefficients and multi-agent systems.

	\appendix
	\begin{proof}
		Since $W$ is a real symmetric matrix, its eigen-subspace of different eigenvalues are orthogonal to each other.
		Denote the orthogonal projection of $\nu$ on the spaces $\mathcal{E}(\lambda_k)$ as
		\begin{equation}
			\nu_k= {\rm Proj} \left( \nu,\mathcal{E}(\lambda_k) \right), \text{ } k=1,2,...,h.
		\end{equation}
		It is obvious that
		\begin{align}
			\nu &= \sum_{k=1}^{h} \nu_k, \label{sum} \\
			\nu_{k_1}^T\nu_{k_2} &=0, \text{ } k_1 \neq k_2, \text{ } 1 \leq k_1,k_2 \leq h. \label{ortho}
		\end{align}
		From (\ref{sum}), (\ref{ortho}) and the condition $\nu_1=0$, we have
		\begin{align}
			\nu^T W \nu 
			&=\sum_{k=2}^{h} \nu_k^T W \sum_{k=2}^{h} \nu_k 
			=\sum_{k=2}^{h} \nu_k^T \sum_{k=2}^{h} \lambda_k \nu_k  
			=\sum_{k=2}^{h} \lambda_k \nu_k^T \nu_k   \nonumber \\
			&\geq \lambda_{\min}^+(W) \sum_{k=2}^{h} \nu_k^T \nu_k
			=\lambda_{\min}^+(W)  \nu^T \nu,
			\label{calculate}
		\end{align}
		where $\lambda_{\min}^+(W)$ is the smallest positive eigenvalue of $W$.
	\end{proof}

	%\clearpage
	%调用文件夹中写好的.bib文件
	\section*{References}
	\vspace{-0.5cm}
	\bibliography{reference}

	\begin{IEEEbiography}
		[{\includegraphics[width=1in,height=1.25in,clip,keepaspectratio]{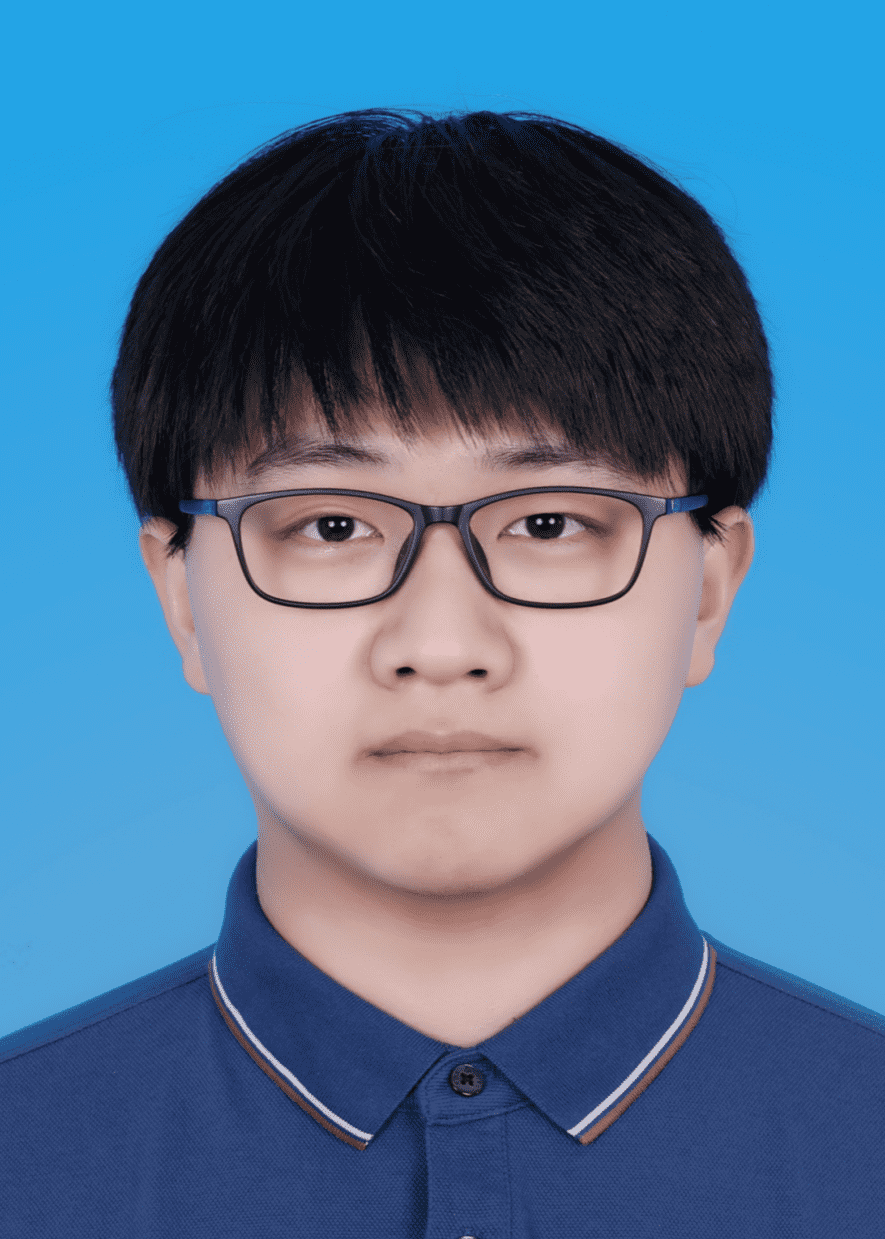}}] 
		{Jiajun Shen} received the B.Eng. degree in automation from Beihang University, Beijing, China, in 2020. 
		He is currently working toward the Ph.D. degree in control theory and control engineering with Beihang University.
		His research interests include data-driven adaptive control, safety-critical control, and distributed cooperative control of multi-agent systems.
		
		He received the Zhang Si-Ying Outstanding Youth Paper Award in the 36th Chinese Control and Decision Conference (CCDC 2024), and the Best Student Paper Award in the 18th International Conference on Control, Automation, Robotics and Vision (ICARCV 2024).
	\end{IEEEbiography}
	
	\begin{IEEEbiography}
		[{\includegraphics[width=1in,height=1.25in,clip,keepaspectratio]{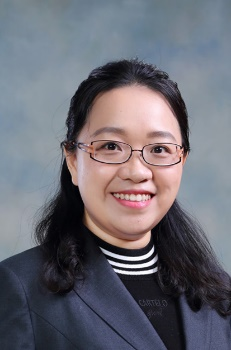}}] 
		{Wei Wang} received her B.Eng degree in Electrical Engineering and Automation from Beihang University (China) in 2005, M.Sc degree in Radio Frequency Communication Systems with Distinction from University of Southampton (UK) in 2006 and Ph.D degree from Nanyang Technological University (Singapore) in 2011. From January 2012 to June 2015, she was a Lecturer with the Department of Automation at Tsinghua University, China. Since July 2015, she has been with the School of Automation Science and Electrical Engineering, Beihang University, China, where she is currently a Full Professor. Her research interests include adaptive control of uncertain systems, distributed cooperative control of multi-agent systems, secure control of cyber-physical systems. 
		
		Prof. Wang received Zhang Si-Ying Outstanding Youth Paper Award in the 25th Chinese Control and Decision Conference (2013), the First Prize of Science and Technology Progress Award by Chinese Institute of Command and Control (2018), and the First Prize of Natural Science Award by Chinese Institute of Simulation (2025). She is the Principle Investigator for a number of research projects including the Distinguished Young Scholars of the National Natural Science Foundation of China (2021-2023). She has been serving as Associate Editors for the IEEE Transactions on Industrial Electronics, ISA Transactions, IEEE Open Journal of Circuits and Systems, Journal of Control and Decision, Journal of Command and Control.  
	\end{IEEEbiography}
	
	\begin{IEEEbiography}
		[{\includegraphics[width=1in,height=1.25in,clip,keepaspectratio]{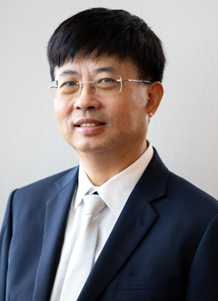}}] 
		{Changyun Wen} received the B.Eng. degree from Xi'an Jiaotong University, China, in 1983 and the Ph.D. degree from the University of Newcastle, Australia in 1990. From August 1989 to August 1991, he was a Postdoctoral Fellow at University of Adelaide, Australia. Since August 1991, he has been with Nanyang Technological University, Singapore, where he is a Full Professor. His main research activities are in the areas of control systems and applications, cyber-physical systems, smart grids, complex systems and networks. As recognition of the scientific impact of his publications in these areas, he was listed as a Highly Cited Researcher by Clarivate for the years of 2020,2021 and 2022.
		
		Prof. Wen is a Fellow of IEEE and Fellow of the Academy of Engineering, Singapore. He was a member of IEEE Fellow Committee from January 2011 to December 2013 and a Distinguished Lecturer of IEEE Control Systems Society from 2010 to 2013. Currently he is the co-Editor-in-Chief of IEEE Transactions on Industrial Electronics, Associate Editor of Automatica (from Feb 2006) and Executive Editor-in-Chief of Journal of Control and Decision. He also served as an Associate Editor of IEEE Transactions on Automatic Control from 2000 to 2002, IEEE Transactions on Industrial Electronics from 2013 to 2020 and IEEE Control Systems Magazine from 2009 to 2019. He has been actively involved in organizing international conferences playing the roles of General Chair (including the General Chair of IECON 2020 and IECON 2023), TPC Chair (e.g. the TPC Chair of Chinese Control and Decision Conference since 2008) ect. 
		
		He was the recipient of a number of awards, including the Prestigious Engineering Achievement Award from the Institution of Engineers, Singapore in 2005, and the Best Paper Award of IEEE Transactions on Industrial Electronics in 2017. 
	\end{IEEEbiography}
	
	\begin{IEEEbiography}
		[{\includegraphics[width=1in,height=1.25in,clip,keepaspectratio]{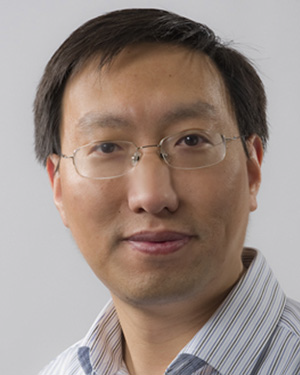}}] 
		{Jinhu L\"u} received the Ph.D. degree in applied mathematics from the Academy of Mathematics and Systems Science, Chinese Academy of Sciences, Beijing, China, in 2002.
		He was a Professor with RMIT University, Melbourne, VIC, Australia, and a Visiting Fellow with Princeton University, Princeton, NJ, USA. 
		Currently, he is a Vice-President for scientific research with Beihang University, Beijing, China. 
		He is/was the Chief Scientist of the National Key Research and Development Program of China and the Leading Scientist of the Innovative Research Groups, National Natural Science Foundation of China. 
		His current research interests include cooperation control, industrial internet, complex networks, and big data.
		
		Dr. L\"u was a recipient of the prestigious Ho Leung Ho Lee Foundation Award in 2015, the National Innovation Competition Award in 2020, the State Natural Science Award three times from the Chinese Government in 2008, 2012, and 2016, respectively, the Australian Research Council Future Fellowships Award in 2009, the National Natural Science Fund for Distinguished Young Scholars Award, and the Highly Cited Researcher Award in engineering from 2014 to 2020, and 2022 to 2023. 
		He was the General Co-Chair of the 43rd Annual Conference of the IEEE Industrial Electronics Society in 2017. 
		He is/was an Editor in various ranks for 16 SCI journals, including the Co-Editor-in-Chief of IEEE Trans. Ind.Inform.
		He is the Fellow of IEEE/CAA/ORSC/CICC.
	\end{IEEEbiography}

\end{document}